\documentclass[a4paper,final,leqno]{siamltex}
\usepackage{enumerate}
\usepackage{amsmath}
\usepackage{amssymb}
\usepackage{amsfonts}
\usepackage{showkeys}

\usepackage{mathrsfs}
\newtheorem{assumption}{Assumption}
\newtheorem{remark}[theorem]{Remark}

\numberwithin{equation}{section}
\numberwithin{theorem}{section}
\newcommand\itemb{\item[$\bullet$]}
\title{Jump-Diffusion Risk-Sensitive Asset Management II: Jump-Diffusion Factor Model\thanks{Research supported by EPSRC under Mathematics and Industry  Grant EP/F035578/1. We wish to express our gratitude to the participants at the Stochastics, Control and Finance Workshop held in London on April 12-14, 2010,  the sixth World Congress of the Bachelier Finance Society held in Toronto on June 22-26, 2010, the ``Modern trends in controlled stochastic processes'' workshop held in Liverpool on July 12-16, 2010 and the 12th Conference on Stochastic Programming (SPXII) held in Halifax on August 16-20, 2010. We would also like to offer our thanks to Mete Soner for his insightful advice and kind encouragement, and to an anonymous referee for detailed comments on the paper which have helped us improve it greatly both in terms of accuracy and readability.}}
\author{Mark Davis\footnote{Department of Mathematics, Imperial College London, London SW7 2AZ, England, Email: mark.davis@imperial.ac.uk} and S\'ebastien Lleo\footnote{Finance Department and Value \& Persuasion Research Centre, Reims Management School, 59 rue Pierre Taittinger, 51100 Reims, France, Email: sebastien.lleo@reims-ms.fr}}
\date{23 July 2012}
\begin{document}
\bibliographystyle{plain}
\maketitle

\begin{abstract}
In this article we extend earlier work on the jump-diffusion risk-sensitive asset management problem in a factor model [SIAM J. Fin. Math. {\bf 2} (2011) 22-54] by allowing jumps in both the factor process and the asset prices, as well as stochastic volatility and investment constraints. In this case, the HJB equation is a partial integro-differential equation (PIDE). We are able to show that finding a viscosity solution to this PIDE is equivalent to finding a viscosity solution to a related PDE, for which classical results give uniqueness. With this in hand, a policy improvement argument and classical results on parabolic PDEs show that the HJB PIDE admits a unique smooth solution. The optimal investment strategy is given by the feedback control that minimizes the Hamiltonian function appearing in the HJB PIDE.
\end{abstract}

\\


\section{Introduction}
Risk-sensitive control generalizes classical stochastic control by considering directly the optimizing agent's degree of risk aversion. In risk-sensitive control, the decision maker's objective is to select a control policy $h(t)$ to maximize the criterion
\begin{equation}\label{eq_criterion_J}
    J(\theta, h) := -\frac{1}{\theta}\ln\mathbf{E}\left[e^{-\theta F_T(h)} \right]
\end{equation}
where $F_T(h)$ is a reward function at a fixed final time $T$ corresponding to a control process $h$, and the exogenous parameter $\theta>0$ represents the decision maker's degree of risk aversion. A Taylor expansion of $J$ around $\theta = 0$ gives
\begin{equation}\label{eq_Taylor_RSC}
    J(\theta,h)
    = \mathbf{E}\left[F_T(h)\right]
    - \frac{\theta}{2}\mathbf{Var}\left[F_T(h)\right]
    + O(\theta^2).
\end{equation}
In optimal investment problems we take $F_T=\log V_T$, where $V_T$ is the value of the investment portfolio corresponding to an asset allocation strategy $h$; then \eqref{eq_Taylor_RSC} shows that risk-sensitive control amounts to `dynamic Markowitz': we maximize the expected return subject to a constraint on the variance of returns.
The reader will find a brief survey of previous work in this area in Part I of this paper \cite{dall_JDRSAM_Diff}.
\\

In \cite{dall_JDRSAM_Diff}, as in Bielecki and Pliska \cite{bipl99} and many other papers, the growth rates of the assets are supposed to depend on a random factor process $X_t$, the components of which can be interpreted either as macroeconomic factors or simply as a statistical representation of the uncertainty of asset returns. We modelled $X_t$ as a gaussian diffusion, while the asset prices themselves were modelled as affine jump-diffusions.
Here we present a more general model that includes jumps in the factor process as well as stochastic volatility and investment constraints.
Including jumps presents more than a mathematical interest. It is also important from a practical perspective as this formulation is necessary to model additional asset classes such as credit products. Our solution technique uses a change-of-measure idea introduced by Kuroda and Nagai \cite{kuna02} that reduces the risk-sensitive optimization problem to a stochastic control problem in the factor process. When, as in \cite{dall_JDRSAM_Diff}, $X_t$ has no jumps, the associated Bellman equation is a partial differential equation (PDE) which was shown to admit a unique classical classical $(C^{1,2})$ solution.
\\

In the framework of the present paper the transformed problem \`a la Kuroda-Nagai is a fully nonlinear controlled jump-diffusion, and the Bellman equation is a is a partial integro-differential equation (PIDE) for which no analytical solution can be expected. In such a situation, viscosity solutions are generally used to show that the value function is the unique continuous solution of the HJB PIDE; see in particular Crandall Ishii and Lions~\cite{crisli92}, Barles and Imbert~\cite{baim08} for an overview of viscosity solutions and Fleming and Soner~\cite{flso06}, \O ksendal and Sulem~\cite{oksu05} or Touzi~\cite{to02} for a discussion of their application to stochastic control, as well as Davis and Lleo~\cite{dall_JDRSAM_Visc} for a viscosity approach to risk-sensitive asset management. In the context of optimal control, a limitation of viscosity solutions is that they are weak solutions: they do not generally satisfy the smoothness assumption required to obtain a sufficient condition via a classical verification theorem.  
\\

Pham~\cite{ph98}, in the context of linear PIDEs, and more recently Davis, Guo and Wu~\cite{guwu09}, in the context of elliptical PIDEs arising from impulse-control problems, suggested a link between viscosity solutions and classical solutions. Their argument differs from more abstract results on classical solutions to PIDEs (see for example~\cite{mipr94}) in that it is directly tailored to the problem at hand. It is therefore of particular interest for the type of applied problems related to optimal control. In this article, we build on this approach to prove that the risk-sensitive Bellman PIDE admits a unique smooth solution and solves the risk-sensitive asset management problem in a jump-diffusion setting. Our situation is nevertheless different from that encountered by Pham or Davis, Guo and Wu. The difficulty in their approach stems form the resolution of a free boundary problem with one source of nonlinearity: the nonlocal integro-differential term associated with the jumps. In contrast, we have no free boundary value problem but three sources of nonlinearity: a quadratic growth term, the optimization embedded in the HJB PDE and the non-local nonlinear integro-differential term. While we can readily eliminate the first nonlinearity through a change of variable, we must tackle the last two simultaneously. 
\\

The article is organised as follows. We introduce the analytical setting in Section \ref{sec_setting} before formulating the control problem in Section \ref{sec_problem}. The main result, Theorem \ref{theo_JDRSAM_main_result}, is stated in Section \ref{main} and proved in detail in the two subsequent sections: in Sections \ref{classical} and \ref{optimalcontrol}, we respectively address the questions of the existence of classical solution to the HJB PIDE, and a verification theorem establishing optimality of our candidate control. An Appendix gives the proof of Proposition \ref{prop_tildePhi_Lipschitz},
establishing Lipschitz continuity of the value function.\\


%
%

\section{Analytical Setting}\label{sec_setting}
The asset market comprises $m$ risky securities $S_i, \; i=1,\ldots m$ and a money market account process $S_0$. The growth rates of the assets depend on $n$ factors $X_1(t), \ldots, X_n(t)$ which follow the dynamics given in the jump diffusion equation~\eqref{eq_FactorProcess} below.
\\

Let $(\Omega, \mathcal{F},
\mathbb{P})$ be the underlying probability space. On this space is
defined an $\mathbb{R}^M$-valued
$\left(\mathcal{F}_t\right)$-Brownian motion $W(t)$ with components
$W_k(t)$, $k=1,\ldots,M$, and $M:=m+n$. Moreover, let $N$ be a $\left(\mathcal{F}_t\right)$-Poisson point process on $(0,\infty)\times\mathbf{Z}$, independent of $W(t)$, where $(\mathbf{Z},\mathcal{B}_{\mathbf{Z}})$ is a given Borel space. Define
\begin{equation}\label{def_JDRSAM_ZFrak_set}
    \mathfrak{Z} := \left\{ U \in \mathcal{B}_{\mathbf{Z}}, \mathbb{E} \left[N(t,U)\right] < \infty \; \forall t\right\}
\end{equation}
For notational convenience, we fix throughout the paper a set $\mathbf{Z}_0 \in \mathcal{B}_{\mathbf{Z}}$ such that
$\nu(\mathbf{Z} \backslash \mathbf{Z}_0)<\infty$ and define, as in \cite{oksu05} 
\begin{eqnarray}
    &&\bar{N}(dt,dz)
                                                \label{nbar}\\
    &=& \left\{ \begin{array}{ll}
        N(dt,dz) - \hat{N}(dt,dz) = N(dt,dz) - \nu(dz)dt =: \tilde{N}(dt,dz)    &   \textrm{if } z \in \mathbf{Z}_0     \\
        N(dt,dz)                      &   \textrm{if } z \in \mathbf{Z} \backslash \mathbf{Z}_0       \\
    \end{array}\right.
                                                \nonumber
\end{eqnarray}

For $t\in[0,T]$ let $\mathcal{F}_{t}$ be the $\sigma$-field generated by the Brownian motions $W_k(s)$ and Poisson processes $N(A,s)$ for $k=1,\ldots, M$, $A\in\mathcal{B}_{\mathbf{Z}}$ and $0\leq s\leq t$, completed with all null sets of $\mathcal{F}_T$. It is well known that the filtration $\{\mathcal{F}_{t}\}_{t\in[0,T]}$ satisfies the `usual conditions'.

\subsection{Factor Dynamics}\label{Chapter3_JDRSAM_theory_factordynamics}
The factor process $X(t)\in\mathbb{R}^{n}$ is allowed to have a full jump-diffusion dynamics,
satisfying the {\sc SDE}
\begin{equation}\label{eq_FactorProcess}
    dX(t) = b\left(t,X(t^-)\right)dt + \Lambda(t,X(t)) dW(t) + \int_{\mathbf{Z}} \xi\left(t,X(t^-),z\right) \bar{N} (dt,dz),
    \qquad X(0) = x_0\in\mathbb{R}^{n}.
\end{equation}
The standing assumptions are as follows.
\\

\begin{assumption}\label{As_Factors}

\begin{enumerate}[(i)]
\item\label{As_Factors_i} The function $b: [0,T]\times\mathbb{R}^n\to \mathbb{R}^n$ is bounded and Lipshitz continuous
\begin{eqnarray}\label{As_Factor_eq1}
                |b(t,y)-b(s,x)| \leq K_b\left(|t-s|+|y-x|\right)
\end{eqnarray}
for some constant $K_b>0$.

\item\label{As_Factors_ii} the function $\Lambda: [0,T]\times\mathbb{R}^n\to\mathbb{R}^{n \times M}$ is bounded and Lipschitz continuous,
\begin{eqnarray}\label{As_Factor_eq2}
                |\Lambda(t,y)-\Lambda(s,x)| \leq K_{\Lambda}\left(|t-s|+|y-x|\right)
\end{eqnarray}
for some constant $K_{\Lambda}>0$.

\item\label{As_Factors_iii} There exists $\eta_{\Lambda} > 0$ such that
\begin{eqnarray}
                \zeta'\Lambda\Lambda'(t,x)\zeta \geq \eta_{\Lambda} |\zeta|^2
\end{eqnarray}
for all $\zeta \in \mathbb{R}^n$

\item\label{As_Factors_iv} There exists $K_b'>0$ and $K_{\Lambda}' >0$ such that
\begin{eqnarray}\label{As_Factor_eq3}
                |b_t| + |b_x| \leq K_b'
\end{eqnarray}
\begin{eqnarray}\label{As_Factor_eq4}
                |\Lambda_t| + |\Lambda_x| \leq K_{\Lambda}'
\end{eqnarray}

\item\label{As_Factors_v} The function $\xi: [0,T]\times\mathbb{R}^n\times\mathbf{Z}\to\mathbb{R}$ is bounded and Lipshitz continuous, i.e.
\begin{eqnarray}
                |\xi(t,y,z)-\xi(s,x,z)| \leq K_{\xi}\left(|t-s|+|y-x|\right)
\end{eqnarray}
for some constant $K_{\xi}>0$. 

\item\label{As_Factors_vi} The vector valued function $\xi(t,x,z)$ satisfies:
\begin{eqnarray}\label{as_factorjumps_xi_integrable}
   \int_{\mathbf{Z}_0} \lvert\xi(t,x,z)  \rvert \nu(dz) < \infty,
        \qquad  \forall (t,x) \in [0,T]\times\mathbb{R}^n
\end{eqnarray}
and for some constant $c$
\begin{equation} \int_\mathbf{Z}|\xi(t,x,z)|^2\nu(dz)<c(1+|x|)^2.\label{as_factorjumps_xi_sq_int}
\end{equation}
\end{enumerate}
\end{assumption}
The minimal condition on $\xi$ under which the factor equation \eqref{eq_FactorProcess}
is well posed is
\[         \int_{\mathbf{Z}_0} \lvert\xi(t,x,z)  \rvert^2 \nu(dz) < \infty,\]
see Definition II.4.1 in Ikeda and Watanabe~\cite{ikwa81}. However, for this paper it is essential to impose the stronger condition \eqref{as_factorjumps_xi_integrable} in order to connect the viscosity solution of the HJB partial integro-differential equation (PIDE) to the viscosity solution of a related parabolic PDE. The same condition is imposed in Davis-Guo-Wu \cite{guwu09}. Condition \eqref{as_factorjumps_xi_sq_int} is needed in Section \ref{optimalcontrol}.
\newline

\begin{remark}
        Note that~\eqref{As_Factor_eq3} and~\eqref{As_Factor_eq4} follow respectively from~\eqref{As_Factor_eq1} and~\eqref{As_Factor_eq2} when $b$ and $\lambda$ are differentiable. 
\end{remark}

\subsection{Asset Market Dynamics}\label{Chapter3_JDRSAM_theory_assetdynamics}
Let $S_0$ denote the wealth invested in the money market account with dynamics given by the equation:
\begin{equation}\label{eq_JDRSAM_BankAccount}
    \frac{dS_{0}(t)}{S_{0}(t)} = a_0\left(t,X(t)\right)dt,
    \qquad S_0(0) = s_0
\end{equation}
and let $S_{i}(t)$ denote the price at time $t$ of the $i$th security, with $i = 1,
\ldots, m$. The dynamics of risky security $i$ can be expressed as:
\begin{eqnarray}\label{eq_SecurityProcess}
    \frac{dS_{i}(t)}{S_{i}(t^{-})} &=&
        \left[a\left(t,X(t^-)\right)\right]_{i}dt
        + \sum_{k=1}^{N} \Sigma_{ik}(t,X(t)) dW_{k}(t)
        + \int_{\mathbf{Z}}\gamma_i(t,z)\bar{N} (dt,dz),
                                                \nonumber\\
    && S_i(0) = s_i,
    \quad i = 1, \ldots, m
\end{eqnarray}
The standing assumptions are as follows.
\\

\begin{assumption}\label{As_Securities}

\begin{enumerate}[(i)]

\item\label{As_Securities_i} the function $a_0$ defined as $a_0: [0,T]\times\mathbb{R}^n\to\mathbb{R}$ is bounded, of class $C^{1,1}\left( [0,T]\times\mathbb{R}^n\right)$ and is Lipshitz continuous
\begin{eqnarray}
                |a_0(t,y)-a_0(s,x)| \leq K_0\left(|t-s|+|y-x|\right)
\end{eqnarray}
for some constant $K_0>0$. 

\item\label{As_Securities_ii} There exists $K_0'>0$ such that
\begin{eqnarray}
                \bigg|\frac{\partial a_0}{\partial t}\bigg| + |Da_0| \leq K_0'
\end{eqnarray}

\item\label{As_Securities_iii} the function $a: [0,T]\times\mathbb{R}^n \to\mathbb{R}^m$ is bounded and Lipshitz continuous, i.e.
\begin{eqnarray}\label{As_Securities_eq1}
                |a(t,y)-a(s,x)| \leq K_a\left(|t-s|+|y-x|\right)
\end{eqnarray}
for some constant $K_a>0$.

\item\label{As_Securities_iv}  the function $\Sigma: [0,T]\times\mathbb{R}^n\to\in \mathbb{R}^{m \times M}$ is bounded and Lipshitz continuous, i.e.
\begin{eqnarray}\label{As_Securities_eq2}
                |\Sigma(t,y)-\Sigma(s,x)| \leq K_{\Sigma}\left(|t-s|+|y-x|\right)
\end{eqnarray}
for some constant $K_{\Sigma}>0$.

\item\label{As_Securities_v}  There exists $\eta_{\Sigma} > 0$ such that
\begin{eqnarray}\label{H11}
                \zeta'\Sigma\Sigma'(t,x)\zeta \geq \eta_{\Sigma} |\zeta|^2
\end{eqnarray}
for all $\zeta \in \mathbb{R}^m$

\item\label{As_Securities_vi}  There exists $K_a'>0$ and $K_{\Sigma}' >0$ such that
\begin{eqnarray}\label{As_Securities_eq3}
                |a_t| + |a_x| \leq K_a'
\end{eqnarray}
\begin{eqnarray}\label{As_Securities_eq4}
                |\Sigma_t| + |\Sigma_x| \leq K_{\Sigma}'
\end{eqnarray}

\item\label{As_Securities_vii}  the function $\gamma  : [0,T]\times\mathbf{Z}\to\mathbb{R}^m$ is bounded, continuous and satisfies the growth condition
\begin{eqnarray}
                |\gamma(t,z)-\gamma(s,z)| \leq K_{\gamma}\left(|t-s|\right)
\end{eqnarray}
for some constant $K_{\gamma}>0$. 

\item\label{As_Securities_viii} The vector valued function $\gamma(t,z)$ satisfy:
\begin{eqnarray}\label{as_factorjumps_gamma_integrable}
   \int_{\mathbf{Z}_0} \lvert\gamma(t,z)  \rvert^2 \nu(dz) < \infty,
        \qquad  \forall (t,x) \in [0,T]\times\mathbb{R}^n
\end{eqnarray}

\item\label{As_Securities_ix} We also require
\begin{eqnarray}
                |\Lambda\Sigma'(t,y)-\Lambda\Sigma'(s,x)| \leq K_{\Lambda\Sigma}\left(|t-s|+|y-x|\right)
\end{eqnarray}
for some constant $K_{\Lambda\Sigma}>0$
\newline
\end{enumerate}

\end{assumption}

\begin{remark}
        Note that~\eqref{As_Securities_eq3} and~\eqref{As_Securities_eq4} follow respectively from~\eqref{As_Securities_eq1} and~\eqref{As_Securities_eq2} when $b$ and $\lambda$ are differentiable. 
\end{remark}

We need a further assumption relating the jumps in the asset prices and factors.

\begin{assumption}\label{as_JDRSAM_uncorrelatedjumps}
$ \gamma(t,z)\xi'(t,x,z) = 0\quad\forall (t,x,z) \in [0,T]\times\mathbb{R}^n\times\mathbf{S}.$
\end{assumption}
This assumption implies that there are no simultaneous jumps in the factor
process and the asset price process. This imposes some restriction, but appears  essential in the argument below.
\\

\subsection{Portfolio Dynamics}\label{JDRSAM_theory_portfoliodynamics}

The function $\gamma$ appearing in \eqref{eq_SecurityProcess} is assumed to satisfy the following conditions.
\begin{assumption}\label{as_assetjumps_upanddown_1}
Define
\begin{equation}
    \mathbf{S} := \mathrm{supp}(\nu) \in \mathcal{B}_{\textbf{Z}}
                                                \nonumber
\end{equation}
and
\begin{equation}
    \tilde{\mathbf{S}}
    :=  \mathrm{supp}(\nu \circ \gamma^{-1})
    \in \mathcal{B}\left(\mathbb{R}^m\right)
                                                \nonumber
\end{equation}
where $\mathrm{supp}(\cdot)$ denotes the support of the measure, and let $\prod_{i=1}^{m} [\gamma_{i}^{min}, \gamma_{i}^{max}]$ be the smallest closed hypercube containing $\tilde{\mathbf{S}}$. We assume that $\gamma(t,z) \in \mathbb{R}^m$ satisfies
\begin{eqnarray}
        -1 \leq \gamma_{i}^{min} \leq \gamma_{i}(t,z) \leq \gamma_{i}^{max} < +\infty
        , \qquad i =1,\ldots,m
                                                                                        \nonumber
\end{eqnarray}and
\begin{eqnarray}
        \gamma_{i}^{min} < 0 < \gamma_{i}^{max}
        , \qquad i =1,\ldots,m
                                                                                        \nonumber
\end{eqnarray}
for $i = 1, \ldots, m$.
\\
\end{assumption}

Define the set $\mathcal{J}_0$ as
\begin{equation}\label{def_JDRSAM_setJ}
    \mathcal{J}_0 := \left\{h \in \mathbb{R}^m:  -1-h'\psi < 0 \quad \forall \psi \in \tilde{\mathbf{S}}\right\}
\end{equation}
For a given $z\in\mathbf{S}$, the equation $h'\gamma(t,z)= -1$ describes
a hyperplane in $\mathbb{R}^m$. Under Assumption~\ref{as_assetjumps_upanddown_1},
$\mathcal{J}_0$ is a convex subset of $\mathbb{R}^m$ for all $(t,x) \in [0,T] \times \mathbb{R}^n$.
\newline

Let $\mathcal{G}_t := \sigma((S(s), X(s)), 0 \leq s \leq t)$ be the
sigma-field generated by the security and factor processes up to
time $t$.

\begin{definition}\label{def_JDRSAM_controlprocess_h}
    An $\mathbb{R}^m$-valued control process $h(t)$ is in class $\mathcal{H}_0$ if the following conditions are satisfied:
    \begin{enumerate}[(i)]
        \item $h(t)$ is progressively measurable with respect to
        $\left\{ \mathcal{B}([0,t]) \otimes \mathcal{G}_t\right\}_{t \geq
        0}$ and is c\`adl\`ag;
        \item~\label{def_JDRSAM_controlprocess_h_cond3} $h(t)\in{\cal J}_0 \quad\forall t \textrm{ a.s.}$
    \end{enumerate}
\end{definition}
\noindent We note that under Assumption~\ref{as_assetjumps_upanddown_1}, a control process $h(t)$ satisfying (ii) is bounded.
\\

By the budget equation, the proportion invested in the money market account is equal to $h_0(t)=1-\sum_{i=1}^{m} h_i(t)$. This implies that the wealth $V(t)$ of the investor in response to an investment strategy $h(t) \in \mathcal{H}_0$, follows the dynamics

\begin{equation}\label{eq_JDRSAM_V_dynamics}
   \frac{dV(t)}{V(t^-)}
       = \left(a_0\left(t,X(t)\right)\right)dt
      + h'(t)\hat{a}\left(t,X(t)\right)dt
        + h'(t)\Sigma(t,X(t)) dW_t
        + \int_{\mathbf{Z}}h'(t)\gamma(t,z)\bar{N} (dt,dz),
\end{equation}
with $V(0) = v_0$, the initial endowment, and  $\hat{a} := a - a_0\mathbf{1}$, where $\mathbf{1} \in \mathbb{R}^m$ denotes the $m$-element unit column vector. 
\\

\subsection{Investment Constraints}
We consider $r \in \mathbb{N}$ fixed investment constraints expressed in the form
\begin{equation}\label{eq_JDRSAM_constraints}
        \Upsilon' h(t)  \leq \upsilon
\end{equation}
where $\Upsilon \in \mathbb{R}^m \times \mathbb{R}^r$ is a matrix and $\upsilon \in \mathbb{R}^r$ is a column vector. For the constrained control problem to be sensible,  we need $\Upsilon$ and $\upsilon$ to satisfy the following condition:
\begin{assumption}\label{as_JDRSAM_sensible_constraints}
        The system
        \begin{equation}\label{eq_JDRSAM_constraints_ineqsyst}
                \Upsilon' y  \leq \upsilon
                                                                                                \nonumber
        \end{equation}
        for the variable $y \in \mathbb{R}^m$ admits at least two solutions.
\\
\end{assumption}

%

We define the feasible region $\mathcal{J}$ as
\begin{eqnarray}\label{eq_JDRSAM_feasibleregion}
        \mathcal{J} := \left\{ h \in {\cal J}_0: \Upsilon' h  \leq \upsilon \right\}
\end{eqnarray}
The feasible region $\mathcal{J}$ is a a convex subset of $\mathbb{R}^m$ and as a result of Assumption~\ref{as_JDRSAM_sensible_constraints}, $\mathcal{J}$ has at least one interior point.

\subsection{Problem formulation} The class $\mathcal{A}$ of admissible investment strategies is defined as follows.
\begin{definition}\label{def_JDRSAM_admissible_A}
    A control process $h(t)$ is in class $\mathcal{A}$ if the
    following conditions are satisfied:
    
 \noindent(i) $h\in \mathcal{H}$, where
                \begin{eqnarray}\label{def_JDRSAM_investprocess_constrained}
                \mathcal{H} := \left\{ h(t) \in \mathcal{H}_0: h(t)\in{\cal J}\,\,\forall
        t\in[0,T],\mathrm{ a.s.} \right\}
                \end{eqnarray}

\noindent(ii) $\mathbf{E} \chi^h(T)= 1$ where $\chi^h(t)$ is the Dol\'eans exponential defined for $t\in[0,T]$ by
\begin{eqnarray}\label{eq_JDRSAM_Doleansexp_chi}
    \chi^h(t)
    &:=& \exp \left\{ -\theta \int_{0}^{t} h(s)'\Sigma(s,X(s)) dW_s
    -\frac{1}{2} \theta^2 \int_{0}^{t} h(s)'\Sigma\Sigma'(s,X(s))h(s) ds            \right.
                                                            \nonumber\\
   &&   \left.
        +\int_{0}^{t} \int_{\mathbf{Z}} \ln\left(1-G(s,z,h(s))\right) \tilde{N} (ds,dz)
            \right.
                                                            \nonumber\\
   &&   \left.
        +\int_{0}^{t} \int_{\mathbf{Z}} \left\{\ln\left(1-G(s,z,h(s))\right)+G(s,z,h(s))\right\}\nu(dz)ds
    \right\},
                                                                                                                                                                \nonumber\\
\end{eqnarray}
with
\begin{eqnarray}
    G(t,z,h) &=&
        1-\left(1+h'\gamma(t,z)\right)^{-\theta}
\end{eqnarray}
\end{definition}

Recall that our objective is to maximise the risk-sensitive criterion $J(h,v)$ of \eqref{eq_criterion_J} with $F_T=\ln V(T)$.
From \eqref{eq_JDRSAM_V_dynamics} and the general It\^o formula we find that the term $e^{-\theta\ln V(T)}$ can be expressed as
\begin{eqnarray}\label{eq_JDRSAM_eminthetaVt}
    e^{-\theta \ln V(T)} &=&  v_0^{-\theta}
        \exp \left\{ \theta \int_{0}^{T} g(t,X_t,h(t)) dt \right\} \chi^h(T)
\end{eqnarray}
where
\begin{eqnarray}\label{eq_JDRSAM_g_func_def}
    g(t,x,h)
    &=&\frac{1}{2} \left(\theta+1 \right)h'\Sigma\Sigma'(t,x)h
                        - a_0(t,x)
                        - h'\hat{a}(t,x)
                                                            \nonumber\\
   && +\int_{\mathbf{Z}}  \left\{\frac{1}{\theta}
        \left[\left(1+h'\gamma(t,z)\right)^{-\theta}-1\right]
        +h'\gamma(t,z)\mathit{1}_{\mathbf{Z}_0}(z)
        \right\} \nu(dz)
\end{eqnarray}
and the Dol\'eans exponential $\chi^h(T)$ is given by~\eqref{eq_JDRSAM_Doleansexp_chi}.
\\

\begin{remark}\label{prop_JDRSAM_g_bounded_Lipschitz}
        For a given, fixed $h$, the functional $g$ is bounded and Lipschitz continuous in the state variable $x$. This follows easily by boundedness and Lipschitz continuity of the coefficients $a_0$, $a$, $\Sigma$ and $\gamma$.
\\
\end{remark}

For $h\in{\cal A}$ and $\theta>0$ let $\mathbb{P}^h$ be the measure on
$(\Omega,\mathcal{F}_T)$ defined via the Radon-Nikod\'ym derivative
\begin{equation}\label{eq_JDRSAM_RNder_chi}
    \frac{d\mathbb{P}^h}{d\mathbb{P}}
    = \chi^h(T),
\end{equation}
and let $\mathbf{E}^h$ denote the corresponding expectation. Then from \eqref{eq_JDRSAM_eminthetaVt} we see that the criterion $J$ is given by
\begin{equation}\label{Jhv} J(h)=\ln v_0-\frac{1}{\theta}\ln \mathbf{E}^h\left[\exp\left(\theta\int_0^Tg(t,X_t,h(t))dt\right)\right].\end{equation}
Evidently, the value $v_0$ plays no role in the optimization process. Throughout the rest of the paper we normalize to $v_0=1$. Moreover, under $\mathbb{P}^h$,
\begin{equation}
    W_{t}^{h} = W_t + \theta \int_{0}^{t} \Sigma(s,X(s))'h(s) ds
            \nonumber
\end{equation}
is a standard Brownian motion and the $\mathbb{P}^{h}$-compensated Poisson random measure is given by
\begin{eqnarray}
    \int_{0}^{t}\int_{\mathbf{Z}_0}\tilde{N} ^{h}(ds,dz)
&=&     \int_{0}^{t}\int_{\mathbf{Z}_0}N (ds,dz)
    -   \int_{0}^{t}\int_{\mathbf{Z}_0} \left\{1-G(s,X(s),z,h(s))\right\}\nu(dz)ds
                \nonumber\\
&=&     \int_{0}^{t}\int_{\mathbf{Z}_0}N (ds,dz)
    -   \int_{0}^{t}\int_{\mathbf{Z}_0} \left\{\left(1+h'\gamma(s,X(s),z)\right)^{-\theta}\right\} \nu(dz)ds
                \nonumber
\end{eqnarray}
As a result, under $\mathbb{P}^h$ the factor process $X(s), \; 0 \leq s \leq t$ satisfies the SDE:
\begin{equation}\label{eq_JDRSAM_state_SDE}
   \boxed{ dX(s)= f(s,X(s),h(s))ds  + \Lambda(s,X(s)) dW_{s}^{\theta}
           + \int_{\mathbf{Z}} \xi\left(s,X(s^-),z\right)\tilde{N} ^{h}(ds,dz),\quad X(0)=x_0}
\end{equation}
where
\begin{align}\label{eq_JDRSAM_func_f}
        f(t,x,h)
        :=      b(t,x)
                -       \theta\Lambda\Sigma(t,x)'h
        +       \int_{\mathbf{Z}}\xi(t,x,z)\left[
              \left(1+h'\gamma(t,z)\right)^{-\theta}
                - \mathit{1}_{\mathbf{Z}_0}(z)
         \right]\nu(dz)
\end{align}
and $b$ is the $\mathbb{P}$-measure drift of the factor process (see \eqref{eq_FactorProcess}).

\begin{remark}\label{rk_JDRSAM_f_Lipschitz}
The drift function $f$ is Lipschitz continuous with coefficient $K_f = K_b + \theta K_{\Lambda\Sigma}+K_{\xi}K_0$ where $K_0>0$ is a constant.
For a constant control $h$ the state process $X(t)$ is a Markov process with generator
\begin{eqnarray}\label{eq_JDRSAM_generator_X}
   \mathcal{L}u(t,x)
        &:=&
                f(t,x,h)'Du
                +       \frac{1}{2}\textrm{tr}\left(\Lambda\Lambda'(t,X)D^2 u \right)
                                                                                                                                                                        \nonumber\\
        &&
        +       \int_{\mathbf{Z}}
              \left\{u\left(x+ \xi(t,x,z)\right)
            - u(x)
            - \xi(t,x,z)' Du
      \right\} \nu(dz)ds
\end{eqnarray}
\end{remark}
In summary, we have shown that the risk-sensitive asset allocation problem is equivalent to the stochastic control problem of minimizing the cost criterion
\begin{equation}\boxed{ \tilde{J}(h)=\mathbf{E}^h\left[\exp\left(\theta\int_0^Tg(t,X_t,h(t))dt\right)\right]}
\label{Jtilde}\end{equation}
over the control set ${\cal A}$ for a controlled process $X_t$ satisfying (in `weak solution' form) the jump-diffusion SDE \eqref{eq_JDRSAM_state_SDE}. The remainder of the paper is devoted to solving the stochastic control problem \eqref{eq_JDRSAM_state_SDE},\eqref{Jtilde}.

\section{Dynamic programming and the value function}\label{sec_problem}
We will solve the control problem by studying the Hamilton-Jacobi-Bellman (HJB) equation of dynamic programming, which involves embedding the original problem in a family of problems indexed by time-space points $(s,x)$, the starting time and position of the controlled process $X_t$. The description here is in the same spirit as Bouchard and Touzi \cite{BouTou11}.

For fixed $s\in[0,T]$ we define the filtration $\{{\cal F}^s_t, t\in[s,T]\}$ by
\[ {\cal F}^s_t=\sigma\{W_k(r)-W_k(t), N(A,r)-N(A,t), k=1,\ldots,M, A\in\mathcal{B}_\mathbf{Z},s\leq r\leq t\}\]
and note that ${\cal F}^s_t$ is independent of ${\cal F}_t$. $X(t)$ will denote the solution of \eqref{eq_FactorProcess} on $[s,t]$ with initial condition $X(s)=x$ and $\mathbb{P}_{s,x}$ the measure on ${\cal F}^s_T$ such that $\mathbb{P}_{s,x}[X_s=x]=1$. The class of admissible controls ${\cal A}^s$ is defined analogously to  ${\cal A}$ above with $h$ adapted to ${\cal F}^s_t$, leading to a change of measure on ${\cal F}^s_T$ defined by the Radon-Nikod\'ym derivative
\[ \frac{d\mathbb{P}^h_{s,x}}{d\mathbb{P}_{s,x}}=\chi^h_s(T).\]

We will now introduce the following two auxiliary criterion functions under the
measure $\mathbb{P}^{h}_{s,x}$:

\begin{eqnarray}  \tilde{I}(s,x,h)
        &=& \mathbf{E}_{s,x}^{h}\left[\exp \left\{ \theta \int_{s}^{T} g(t,X_t,h(t)) dt\right\}
                        \right]\label{Itilde}\\
I(s,x,h)&=&-\frac{1}{\theta}\ln \tilde{I}(s,x,h).\label{I}\end{eqnarray}

\begin{remark}\label{rem_criterion_tildeI}
The criterion $\tilde{I}$ defined in~\eqref{Itilde}, which is the cost function for our stochastic control problem,
can be interpreted as a payoff of 1 at the terminal time $T$ `discounted' at a stochastic controlled rate of $-\theta g(\cdot)$ (which is however not necessarily $\geq 0$).
\end{remark}

The corresponding value functions are
\begin{equation}
        \tilde{\Phi}(s,x)=\inf_{h\in{\cal A}^s}\tilde{I}(s,x,h);\qquad \Phi(s,x)=\sup_{h\in{\cal A}^s}I(s,x,h).\label{Phi}
\end{equation}

\begin{lemma} $\tilde{\Phi}(s,x)=\inf_{h\in{\cal A}}\tilde{I}(s,x,h)$. That is, the infimum is unchanged if the class ${\cal A}^s$ is replaced by the larger class ${\cal A}$.\end{lemma}

\proof This uses exactly the argument of Remark 2, page 958 of Bouchard and Touzi \cite{BouTou11}. We condition on the initial filtration and use the independence of ${\cal F}_s$ and ${\cal F}^s_t$.\hfill\endproof

\subsection{The Risk-Sensitive Control Problems under $\mathbb{P}_{h}$}\label{sec_Optim}

We will show that the value function $\Phi$ defined in~\eqref{Phi} satisfies the HJB PIDE
\begin{equation}\label{eq_JDRSAM_HJBPDE}
    \frac{\partial \Phi}{\partial t}
    + \sup_{h \in \mathcal{J}}
    L^{h}\left(t, x,\Phi,  D\Phi, D^2\Phi \right) = 0
\end{equation}
where $\mathcal{J}$ is defined in~\eqref{eq_JDRSAM_feasibleregion},
\begin{eqnarray}\label{eq_JDRSAM_HJBPDE_operator_L}
    L^{h}\left(t, x, u, p, M \right)
    &=& f(t,x,h)'p
        + \frac{1}{2} \textrm{tr} \left( \Lambda\Lambda'(t,x) M \right)
        - \frac{\theta}{2} p' \Lambda\Lambda'(t,x) p
                                                                                        \nonumber\\
    &&
        - g(t,x,h)
        + \mathcal{I}_{NL}\left[t,x,u,p \right]
\end{eqnarray}
with 
\begin{eqnarray}
        \mathcal{I}_{NL}\left[t,x,u, p\right]
        &=&  \int_{\mathbf{Z}} \left\{
                - \frac{1}{\theta}\left(
          e^{-\theta\left[u\left(t,x+\xi(t,x,z)\right)- u(t,x) \right]} -1
         \right)
        - \xi(t,x,z)' p
        \right\} \nu(dz)
\end{eqnarray}
and subject to the terminal condition (recall our normalization $v_0=1$)
\begin{equation}\label{eq_JDRSAM_HJBPDE_termcond}
        \Phi(T, x) = 0
        ,\qquad x \in \mathbb{R}^n.
\end{equation}
Condition \eqref{as_factorjumps_xi_integrable} ensures that $\mathcal{I}_{NL}$ is well defined, at least for bounded $u$.

For $\tilde{\Phi}$, the corresponding HJB PIDE is
\begin{eqnarray}\label{eq_JDRSAM_exptrans_HJBPDE}
   &&  \frac{\partial \tilde{\Phi}}{\partial t}(t,x)
                + \frac{1}{2} \textrm{tr} \left( \Lambda\Lambda'(t,x) D^2 \tilde{\Phi}(t,x)\right)
                        + H(t,x,\tilde{\Phi},D\tilde{\Phi})
                                                                                                        \nonumber\\
        &&      + \int_{\mathbf{Z}} \left\{
              \tilde{\Phi}\left(t,x+ \xi(t,x,z)\right)
            - \tilde{\Phi}(t,x)
            - \xi(t,x,z)' D\tilde{\Phi}(t,x)
        \right\} \nu(dz)=0
\end{eqnarray}
subject to terminal condition
\begin{eqnarray}\label{eq_JDRSAM_exptrans_HJBPDE_termcond}
        \tilde{\Phi}(T, x) = 1
\end{eqnarray}
where for $r \in \mathbb{R}$, $p \in \mathbb{R}^n$
\begin{eqnarray}\label{eq_JDRSAM_logtrans_H_function}
   H(s,x,r,p)
        &=& \inf_{h \in \mathcal{J}} \left\{
                        f(s,x,h)'p
                        +       \theta g(s,x,h) r
        \right\}
                                                                                                                \nonumber\\
\end{eqnarray}
\begin{remark}\label{rk_JDRSAM_H_Lipschitz}
The function $H$ satisfies a Lipschitz condition as well as the linear growth condition
\begin{eqnarray}
        \lvert H(s,x,r,p) \rvert
        \leq C\left(1 + \lvert p \rvert\right)
        ,       \quad   \forall (s,x) \in Q_0
                                                        \nonumber
\end{eqnarray}
\end{remark}
The value functions $\Phi$ and $\tilde{\Phi}$ are related through the strictly monotone continuous transformation $\tilde{\Phi}(t,x) = \exp \left\{-\theta \Phi(t,x) \right\}$. Thus an admissible (optimal) strategy for the exponentially transformed problem is also admissible (optimal) for the risk-sensitive problem. In the remainder of the article, we will refer to the control problem and HJB PIDE related to the value function $\Phi$ as the \emph{risk sensitive} control problem and the \emph{risk sensitive} HJB PIDE, and to the control problem and HJB PIDE related to the value function $\tilde{\Phi}$  as the \emph{exponentially transformed} control problem and the \emph{exponentially transformed} HJB PIDE.

\subsection{Properties of the Value Function $\tilde{\Phi}$} We start by establishing two \emph{a priori} properties of the value function.
\begin{proposition}\label{prop_JDRSAM_tildePhi_bounded}
There exists $M>0$ such that
\[ 0 < \tilde{\Phi}(t,x) \leq M \qquad \forall (t,x) \in [0,T]\times\mathbb{R}^n.\]
\end{proposition} 
\proof
The strategy of investing only in the money-market account, i.e. taking $h\equiv 0$ is sub-optimal, and hence
\[
    \tilde{\Phi}(t,x)
    \leq  \mathbf{E}_{t,x}^{0}e^{\theta \int_{t}^{T} g(X(s),0)ds}
    = \mathbf{E}_{t,x}^{0} e^{\theta \int_{t}^{T} a_0(s,X(s)ds}
    \leq e^{\theta\check{a}_0(T-t)},\]                                                                     
where $\check{a}_0$ is a bound for $|a_0(t,x)|$ (see Assumption \ref{As_Securities}(\ref{As_Securities_i})).
\newline

Moreover,
\begin{eqnarray}
        \tilde{\Phi}(t,x)
        &=&     \inf_{h \in \mathcal{A}} \mathbf{E}_{t,x}^{h}
                        \left[ \exp \left\{ \theta \int_{t}^{T} g(s,X_s,h(s))ds
                        \right\} \right]
        > 0
                                                                                                \nonumber
\end{eqnarray}
This follows from Corollary~\ref{coro_JDRSAM_optimcontrol_tilde} below: the concave minimization problem admits a unique minimizer which is an interior point of the set $\mathcal{J}$ defined in equation~\eqref{def_JDRSAM_setJ}. This concludes the proof.\hfill\endproof

\begin{proposition}\label{prop_tildePhi_Lipschitz}
        The value function $\tilde{\Phi}$ is Lipschitz continuous in the state variable $x$.
\end{proposition}

\begin{proof}
        See Appendix~\ref{App_Proof_prop_tildePhi_Lipschitz}.
\end{proof}

\begin{proposition}\label{prop_JDRSAM_optimcontrol}
        Under (\ref{H11}) and Assumption~\ref{as_JDRSAM_uncorrelatedjumps}, the supremum in~\eqref{eq_JDRSAM_HJBPDE}, \eqref{eq_JDRSAM_HJBPDE_operator_L}  admits a unique Borel measurable maximizer $\hat{h}(t,x,p)$ for $(t,x,p) \in [0,T]\times\mathbb{R}^n\times\mathbb{R}^n$.
\end{proposition}

\begin{proof}
The supremum in~\eqref{eq_JDRSAM_HJBPDE} can be expressed as
\begin{eqnarray}\label{eq_JDRSAM_supL_deriv}
        && \sup_{h \in \mathcal{J}}  L^{h}\left(t, x, u, p, M \right)
                    \nonumber\\
        &=& \sup_{h \in \mathcal{J}} \left\{
        \left( b(t,x)
            +\int_{\mathbf{Z}} \xi(t,x,z)\left[
                \left( 1+h'\gamma(t,z)\right)^{-\theta}
                    - \mathit{1}_{\mathbf{Z}_0}(z)
                \right]\nu(dz)
            \right)'p
            -\theta h'\Sigma\Lambda(t,x)'p
            \right. \nonumber\\
            &&\left.
            + \frac{1}{2} \textrm{tr} \left( \Lambda\Lambda'(t,x)' M \right)
            - \frac{\theta}{2} p'\Lambda\Lambda'(t,x)' p
            + \mathcal{I}_{NL}\left[t,x,u,p \right]
                \right. \nonumber\\
                &&\left.
            -\frac{1}{2} \left(\theta+1 \right)h'\Sigma\Sigma'(t,x)h
            + a_0(t,x)
           + h'\hat{a}(t,x)
                \right. \nonumber\\
            &&\left.
            - \int_{\mathbf{Z}}  \left\{\frac{1}{\theta}
                \left[\left(1+h'\gamma(t,z)\right)^{-\theta}-1\right]
                +h'\gamma(t,z)\mathit{1}_{\mathbf{Z}_0}(z)
            \right\} \nu(dz)
            \right\}
                    \nonumber\\
        &=&
                        b'(t,x)p
            +   \frac{1}{2} \textrm{tr} \left( \Lambda\Lambda'(t,x) M \right)
            - \frac{\theta}{2} p'\Lambda\Lambda'(t,x)' p
            +   a_0(t,x)
            + \mathcal{I}_{NL}\left[t,x,u,p \right]
                \nonumber\\
        &&
            +\sup_{h \in \mathcal{J}} \left\{
            - \frac{1}{2} \left(\theta+1 \right)h'\Sigma\Sigma'(t,x)'h
            -\theta h'\Sigma\Lambda'(t,x)p
            +h'\hat{a}(t,x)
                \right. \nonumber\\
            &&\left.
            -\frac{1}{\theta}\int_{\mathbf{Z}}\left\{
                        \left(1 - \theta\xi(t,x,z)'p\right)
                                                \left[\left(1+h'\gamma(t,z)\right)^{-\theta}-1\right]
                                        +\theta h'\gamma(t,z)\mathit{1}_{\mathbf{Z}_0}(z)
                    \right\}\nu(dz)
            \right\}
\end{eqnarray}

Define the auxiliary functional
\begin{eqnarray}
        \ell(h;x,p)
        &=&
                \frac{1}{2} \left(\theta+1 \right)h'\Sigma\Sigma'(t,x)h
            +\theta h'\Sigma\Lambda'(t,x)'p
            -h'\hat{a}(t,x)
                                                                                                                        \nonumber\\
        &&
                +\frac{1}{\theta}\int_{\mathbf{Z}}\left\{
                \left(1 - \theta\xi(t,x,z)'p \right)
                                \left[\left(1+h'\gamma(t,z)\right)^{-\theta}-1\right]
                +\theta h'\gamma(t,z)\mathit{1}_{\mathbf{Z}_0}(z)
      \right\}\nu(dz)
                                                                                                                                                \nonumber
\end{eqnarray}
for $h \in \mathbb{R}^m$, $x \in \mathbb{R}^n$, $p \in \mathbb{R}^n$ and $\theta \in (0,\infty)$. Under Assumption (\ref{H11}), for any $p \in \mathbb{R}^n$ the terms
\begin{eqnarray}
    \frac{1}{2} \left(\theta+1 \right)h'\Sigma\Sigma'(t,x)h
    + \theta h'\Sigma\Lambda'(t,x)p
    - h'\hat{a}(t,x)
    + \int_{\mathbf{Z}} h'\gamma(t,z)\mathit{1}_{\mathbf{Z}_0}(z)\nu(dz)
                                        \nonumber
\end{eqnarray}
is strictly convex in $h$ $\forall (t,x,z) \in [0,T]\times\mathbb{R}^n\times\mathbf{Z}$ a.s. $d\nu$. Under Assumption~\ref{as_JDRSAM_uncorrelatedjumps}, the nonlinear jump-related term
\begin{eqnarray}
        \frac{1}{\theta}\int_{\mathbf{Z}}\left\{
        \left(1 - \theta\xi'(t,x,z)p\right)
           \left[\left(1+h'\gamma(t,z)\right)^{-\theta}-1\right]
    \right\}\nu(dz)
                                                                                                                                \nonumber
\end{eqnarray}
simplifies to
\begin{eqnarray}
        \frac{1}{\theta}\int_{\mathbf{Z}}
       \left\{
               \left[\left(1+h'\gamma(t,z)\right)^{-\theta}-1\right]
           \right\}\nu(dz)
                                                                                                                                \nonumber
\end{eqnarray}
which is also convex in $h$ $\forall (t,x,z) \in [0,T]\times\mathbb{R}^n\times\mathbf{Z}$ a.s. $d\nu$.
\\

As a function of the variable $h$, $\ell(h;x,p)$ can be defined more precisely as a mapping from the vector space $\mathbb{R}^m$ into $\mathbb{R}$. Moreover, $\ell$ is continuous in $h$ $\forall h \in \mathbb{R}^m$, twice differentiable and with continuous derivatives. Finally, $f$ attains its infimum, and the infimum is finite.
\\

Looking at the constraints, the matrix $\Upsilon$ defines a mapping from the vector space $\mathbb{R}^m$ into the normed space generated by associating to the constraint vector space $\mathcal{U}$ the Euclidian norm. Under Assumption~\ref{as_JDRSAM_sensible_constraints}, there exists an $h_1$ such that $\Upsilon ' h < \upsilon$.
\\

As a result, the auxiliary constrained optimization problem
\begin{eqnarray}
        \min_{h \in \mathcal{U}}        \ell(h;x,p)
                                                                                                                                        \nonumber
\end{eqnarray}
is a convex programming problem satisfying the assumptions of Lagrange Duality (see for example Theorem 1 in Section 8.6 in~\cite{lu69}). We therefore conclude that the supremum is reached for a unique maximizer $\hat{h}(t,x,p)$, which is an interior point of the set $\mathcal{J}$ defined in equation~\eqref{def_JDRSAM_setJ}, and the supremum, evaluated at $\hat{h}(t,x,p) \in \mathbb{R}^n$, is finite. By measurable selection, $\hat{h}$ can be taken as a Borel measurable function on $[0,T]\times\mathbb{R}^n\times\mathbb{R}^n$.
\\
\end{proof}

\begin{corollary}\label{coro_JDRSAM_optimcontrol_tilde}
        Under (\ref{H11}) and Assumption~\ref{as_JDRSAM_uncorrelatedjumps}, the infimum in~\eqref{eq_JDRSAM_logtrans_H_function} admits a unique Borel measurable minimizer $\check{h}(t,x,r,p)$ for $(t,x,r,p) \in [0,T]\times\mathbb{R}^n\times\mathbb{R}\times\mathbb{R}^n$.
\end{corollary}

\subsection{Main result}\label{main}
We now come to the main result of this paper.

\begin{theorem}\label{theo_JDRSAM_main_result}
Under Assumptions~\ref{As_Factors}--\ref{as_JDRSAM_sensible_constraints}, the following hold:
\begin{enumerate}[1.]
\item The exponentially transformed value function $\tilde{\Phi}$ defined at \eqref{Phi} is the unique $C^{1,2}\left([0,T]\times \mathbb{R}^n\right)$ solution  of the RS HJB PIDE~\eqref{eq_JDRSAM_exptrans_HJBPDE}-\eqref{eq_JDRSAM_exptrans_HJBPDE_termcond}.

\item The value function $\Phi$, also defined at \eqref{Phi}, is the unique $C^{1,2}\left([0,T]\times \mathbb{R}^n\right)$ solution  of the RS HJB PIDE~\eqref{eq_JDRSAM_HJBPDE}-\eqref{eq_JDRSAM_HJBPDE_termcond}.

\item The asset allocation $h^*(t) = \hat{h}(t,X_t,D\Phi(t,X_t))$, where $\hat{h}$ is the function introduced in Proposition \ref{prop_JDRSAM_optimcontrol}, is optimal in the class ${\cal A}$ of admissible controls.
\end{enumerate}
\end{theorem}

The detailed argument for the proof of this theorem is given in the next two sections of the paper. Section \ref{classical} establishes the existence of a unique $C^{1,2}$ solution to the HJB equation. Section \ref{optimalcontrol} shows that the function $\hat{h}$ of Proposition \ref{prop_JDRSAM_optimcontrol}  provides an investment strategy that is both admissible and optimal. Combining the results of these two sections, we have the following proof.

\begin{proof}[Proof of Theorem~\ref{theo_JDRSAM_main_result}]

\emph{Existence of a classical ($C^{1,2}$) solution - } by Theorem~\ref{Theo_JDRSAM_existence_PIDE}, the functions $\tilde{\Phi}$ and $\Phi$ are, respectively, the unique $C^{1,2}([0,T]\times \mathbb{R}^n)$ solutions of the HJB PIDE~\eqref{eq_JDRSAM_exptrans_HJBPDE}--\eqref{eq_JDRSAM_exptrans_HJBPDE_termcond} and HJB PIDE~\eqref{eq_JDRSAM_HJBPDE}--\eqref{eq_JDRSAM_HJBPDE_termcond}.

\emph{Existence of an optimal control - } by Proposition~\ref{prop_JDRSAM_optimcontrol}, the supremum in~\eqref{eq_JDRSAM_HJBPDE} and infimum in \eqref{eq_JDRSAM_exptrans_HJBPDE} admit the same unique Borel measurable maximizer/minimizer $h^*(t,X_t)$. By Proposition~\ref{prop_JDRSAM_hstar_admissible}, the control $h^*$ defined by $h^*(t, X(t))$ is admissible, i.e. belongs to the class ${\cal A}$. Theorem \ref{Theo_optimalcontrol} shows by a martingale argument that this control is optimal. 
\hfill\end{proof}

%
%

\section{Existence of a Classical ($C^{1,2}$) Solution}\label{classical}
The objective of this section is to prove that the value functions $\Phi$ and $\tilde{\Phi}$ are smooth. The process involves 6 steps, which we give in outline here and in detail in the six succeeding sections, \S\S4.1-4.6.

\smallskip\noindent{\it Step 1: $\tilde{\Phi}$ is a Lipschitz Continuous Viscosity Solution (VS-PIDE) of~\eqref{eq_JDRSAM_exptrans_HJBPDE}.}
Theorem \ref{theo_JDRSAM_viscositysol} below asserts that $\tilde{\Phi}$ is a (possibly discontinuous) viscosity solution of the PIDE \eqref{eq_JDRSAM_exptrans_HJBPDE}. However, we know \emph{a priori} from Proposition \ref{prop_tildePhi_Lipschitz} that $\tilde{\Phi}$ is Lipschitz.

\smallskip\noindent{\it Step 2: From PIDE to PDE.} 
At this point we invoke Assumption 1(vi) \eqref{as_factorjumps_xi_integrable}, which implies that we can write the non-local term in \eqref{eq_JDRSAM_exptrans_HJBPDE} as 
\[ \int_{\mathbf{Z}}\{\tilde{\Phi}(t,x+\xi(t,x,z))-\tilde{\Phi}(t,x)\}\nu(dz)+\int_{\mathbf{Z}}\xi'(t,x,z)\nu(dz)D\tilde{\Phi}(t,x).\]
Change notation and rewrite the HJB PIDE as a parabolic PDE \textit{\`a la} Pham~\cite{ph98}:
\begin{equation}\label{eq_JDRSAM_exptrans_HJBPDE_parabolic_outline}
                \frac{\partial u}{\partial t}(t,x)
                + \frac{1}{2} \textrm{tr} \left( \Lambda\Lambda'(t,x) D^2 u\right) + H_{a}(t,x,u,Du) + d_{a}^{\tilde{\Phi}}(t,x)= 0
\end{equation}
subject to terminal condition $u(T, x) = 1$ and with
\begin{eqnarray}\label{eq_JDRSAM_logtrans_H_function_parabolic_outline}
   H_{a}(s,x,r,p) &=& \inf_{h \in \mathcal{U}} \left\{
                        f_{a}(s,x,h)'p
                        +       \theta g(s,x,h) r
         \right\}
\end{eqnarray}
for $r \in \mathbb{R}$, $p \in \mathbb{R}^n$ and where
\begin{eqnarray}\label{def_JDRSAM_function_tilde_f_outline}
        f_{a}(s,x,h)
        &:=&    f(s,x,h) - \int_{\mathbf{Z}} \xi(s,x,z) \nu(dz)\nonumber\\
        &=&     b(s,x)-\theta\Lambda\Sigma(s,x)'h
        +       \int_{\mathbf{Z}}\xi(s,x,z)\left[
              \left(1+h'\gamma(s,z)\right)^{-\theta}
               -        \mathit{1}_{\mathbf{Z}_0}(z)
                                        -       1
         \right]\nu(dz),
\end{eqnarray}
and
\begin{eqnarray}\label{def_JDRSAM_function_tilde_d_outline}
        d_{a}^{\tilde{\Phi}}(t,x)
        =       \int_{\mathbf{Z}} \left\{
              \tilde{\Phi}\left(t,x+\xi(t,x,z)\right)
            - \tilde{\Phi}(t,x)
        \right\} \nu(dz).
\end{eqnarray}

\smallskip\noindent{\it Step 3: Viscosity Solution to PDE~\eqref{eq_JDRSAM_exptrans_HJBPDE_parabolic_outline}.} We consider viscosity solutions $u$ of
the semi-linear PDE~\eqref{eq_JDRSAM_exptrans_HJBPDE_parabolic_outline} (always interpreted as an equation for `unknown' $u$ with the last term prespecified, with $\tilde{\Phi}$ defined as in Step 1.) The key point is that \emph{$\tilde{\Phi}$ is a  viscosity solution of the PDE~\eqref{eq_JDRSAM_exptrans_HJBPDE_parabolic_outline}}. Indeed, due to definition \eqref{def_JDRSAM_function_tilde_d_outline},  PIDE~\eqref{eq_JDRSAM_exptrans_HJBPDE} and PDE~\eqref{eq_JDRSAM_exptrans_HJBPDE_parabolic_outline} are in essence the same equation. Hence, if $\tilde{\Phi}$ satisfies the PIDE in the viscosity sense, which from Step 1 we know that it does, then $u=\tilde{\Phi}$ is a viscosity solution of the PDE \eqref{eq_JDRSAM_exptrans_HJBPDE_parabolic_outline}. Note that this last statement depends crucially on the Definition \ref{def_JDRSAM_viscositysol_testfunc_alt} of `viscosity solution' for the PIDE (i.e. no replacement of the solution by a test function in the non-local term.)

\smallskip\noindent{\it Step 4: Uniqueness of the Viscosity Solution to the PDE~\eqref{eq_JDRSAM_exptrans_HJBPDE_parabolic_outline}.} 
If a function $u$ solves the PDE~\eqref{eq_JDRSAM_exptrans_HJBPDE_parabolic_outline} it does not mean that $u$ also solves the PIDE~\eqref{eq_JDRSAM_exptrans_HJBPDE} because the term $d_a$ in the PDE~\eqref{eq_JDRSAM_exptrans_HJBPDE} depends on $\tilde{\Phi}$ regardless of the choice of $u$. Thus, if we were to show the existence of a classical solution $u$ to PDE~\eqref{eq_JDRSAM_exptrans_HJBPDE_parabolic_outline}, we would not be sure that this solution is the value function $\tilde{\Phi}$ unless we can show that PDE~\eqref{eq_JDRSAM_exptrans_HJBPDE_parabolic_outline} admits a unique solution. This only requires applying a ``classical'' comparison result for viscosity solutions (see Theorem 8.2 in Crandall, Ishii and Lions~\cite{crisli92}) provided appropriate conditions on $f_a$ and $d_a^{\tilde{\Phi}}$ are satisfied.

\smallskip\noindent{\it Step 5: Existence of a Classical Solution to the HJB PDE~\eqref{eq_JDRSAM_exptrans_HJBPDE}.}
We use an argument similar to that of Fleming and Rishel~\cite{fl75} (Appendix E) together with a result from Davis, Guo and Wu~\cite{guwu09} to show the existence of a classical solution to the PDE~\eqref{eq_JDRSAM_exptrans_HJBPDE_parabolic_outline} with $d_{a}^{\tilde{\Phi}}(t,x)$ regarded as an autonomous term.

\smallskip\noindent{\it Step 6: Existence of a Classical Solution to the HJB PIDE~\eqref{eq_JDRSAM_exptrans_HJBPDE}.} Combining Steps 4 and 5, we conclude that $\tilde{\Phi}$ and $\Phi$ are respectively a classical ($C^{1,2}$) solution of~\eqref{eq_JDRSAM_exptrans_HJBPDE} and a classical ($C^{1,2}$) solution of~\eqref{eq_JDRSAM_HJBPDE}.

\subsection{$\tilde{\Phi}$ is a Lipschitz Continuous Viscosity Solution of~\eqref{eq_JDRSAM_exptrans_HJBPDE}}

\subsubsection{Preliminary Definitions}

The theory of viscosity solutions applies to elliptical PIDEs of the form
\begin{equation}
    F(t,x,u,Du,D^2u,\mathcal{I}[t,x,u]) = 0
                                            \nonumber
\end{equation}
where $\mathcal{I}[t,x,u]$ is the nonlocal operator, and parabolic PIDEs of the form
\begin{equation}\label{eq_JDRSAM_viscosity_general parabolic_PDE}
    \frac{\partial u}{\partial t} + F(t,x,u,Du,D^2u,\mathcal{I}[t,x,u]) = 0
                                            \nonumber
\end{equation}
for a ``proper'' functional $F(t,x,r,p,A,l)$.

\begin{definition}\label{def_JDRSAM_proper}
A functional $F(t,x,r,p,A,l)$ is said to be \emph{proper} if it satisfies the following two properties:
\begin{enumerate}
    \itemb (degenerate) ellipticity:
    \begin{equation}
        F(t,x,r,p,A,l_1) \leq F(t,x,r,p,B,l_2),
        \qquad B \leq A, l_1 \leq l_2
                                    \nonumber
    \end{equation}
    and
    \itemb monotonicity
    \begin{equation}
        F(t,x,r,p,A,l) \leq F(t,x,s,p,A,l),
        \qquad r \leq s
                                    \nonumber
    \end{equation}
\end{enumerate}

\end{definition}

We now give a definitions of viscosity solutions based on the notion of test function adapted from Barles and Imbert (see Definition 1 in~\cite{baim08}):

\begin{definition}[Viscosity Solution (Test Functions)]\label{def_JDRSAM_viscositysol_testfunc}
A bounded function $u \in USC([0,T]\times\mathbb{R}^n)$ is a \emph{viscosity subsolution} of~\eqref{eq_JDRSAM_exptrans_HJBPDE_viscosity}, if for any bounded test function $\psi \in C^{1,2}([0,T]\times\mathbb{R}^n)$, if $(x,t)$ is a global maximum point of $u-\psi$, then
\begin{eqnarray}
        - \frac{\partial \psi}{\partial t}
        - F\left(t,x,u(t,x),D\psi(t,x),D^2\psi(t,x),\mathcal{I}[t,x,\psi]\right)
        \leq 0
\end{eqnarray}

A bounded function $u \in LSC([0,T]\times\mathbb{R}^n)$ is a \emph{viscosity supersolution} of~\eqref{eq_JDRSAM_exptrans_HJBPDE_viscosity}, if for any bounded test function $\psi \in C^{1,2}([0,T]\times\mathbb{R}^n)$, if $(x,t)$ is a global minimum point of $u-\psi$, then
\begin{eqnarray}
        - \frac{\partial \psi}{\partial t}
        - F\left(t,x,u(t,x),D\psi(t,x),D^2\psi(t,x),\mathcal{I}[t,x,\psi]\right)
        \geq 0
\end{eqnarray}

A bounded function $u$ whose upper semicontinuous and lower semicontinuous envelopes are a viscosity subsolution and a viscosity supersolution of~\eqref{eq_JDRSAM_HJBPDE} is a viscosity solution of~\eqref{eq_JDRSAM_HJBPDE}.

\end{definition}

An equivalent definition of viscosity solutions based on the notion of semijets (see Definition 4 and Proposition 1 in Barles and Imbert~\cite{baim08}) is crucial to extend Ishii's lemma to account for non-local operators, obtain a stability result and also show uniqueness of the viscosity solution. In Section~\ref{sec_classicalsol}, we will consider a third equivalent definition which will enable us to rewrite the PIDE~\eqref{eq_JDRSAM_exptrans_HJBPDE} as a `parabolic' PDE.
\\

\subsubsection{Characterization of the Value Function as a Viscosity Solution}
With regards to our problem, we can express the HJB PIDE~\eqref{eq_JDRSAM_exptrans_HJBPDE} associated with the exponentially transformed problem~\eqref{Phi} as
\begin{eqnarray}\label{eq_JDRSAM_exptrans_HJBPDE_viscosity}
   - \frac{\partial \tilde{\Phi}}{\partial t}(t,x)
        + F_{v}(t,x,\tilde{\Phi},D\tilde{\Phi},D^2\tilde{\Phi},\mathcal{I}[t,x,\tilde{\Phi}])
   &=& 0
\end{eqnarray}
subject to terminal condition~\eqref{eq_JDRSAM_exptrans_HJBPDE_termcond}, where
\begin{eqnarray}\label{eq_JDRSAM_exptrans_functional_F}
F_{v}(t,x,\tilde{\Phi},D\tilde{\Phi},D^2\tilde{\Phi},\mathcal{I}[t,x,\tilde{\Phi}])
   &:=&
        H_{v}(t,x,\tilde{\Phi},D\tilde{\Phi})
        -       \frac{1}{2} \textrm{tr} \left( \Lambda\Lambda'(t,x) D^2 \tilde{\Phi}(t,x)\right)
        -       \mathcal{I}[t,x,\tilde{\Phi}]
                                                                                                        \nonumber
\end{eqnarray}
\begin{eqnarray}\label{eq_JDRSAM_nonlocaloperator}
        \mathcal{I}[t,x,\tilde{\Phi}]
        &:=& \int_{\mathbf{Z}} \left\{
              \tilde{\Phi}\left(t,x+ \xi(t,x,z)\right)
            - \tilde{\Phi}(t,x)
            - \xi(t,x,z)' D\tilde{\Phi}(t,x)\mathit{1}_{\mathbf{Z}_0}
        \right\} \nu(dz)
\end{eqnarray}
\begin{eqnarray}\label{eq_JDRSAM_logtrans_H_tilde_function}
   H_{v}(s,x,r,p)
        &:=&     - H(s,x,r,p)
                                                                                                        \nonumber\\
        &=& \sup_{h \in \mathcal{A}} \left\{
                -       f_{v}(s,x,h)'p
                        -       \theta g(s,x,h) r
        \right\}
                                                                                                                \nonumber
\end{eqnarray}
for $r \in \mathbb{R}$, $p \in \mathbb{R}^n$, and where
\begin{eqnarray}\label{eq_JDRSAM_func_fvisc}
        f_{v}(t,x,h)
        &:=&    f(t,x,h)
                -       \int_{\mathbf{Z}\backslash\mathbf{Z}_{\delta}}\xi(t,x,z)\nu(dz)
                                                                                                                \nonumber\\
        &=&     b(t,x)
                -       \theta\Lambda\Sigma(t,x)'h(s)
        +       \int_{\mathbf{Z}}\xi(t,x,z)\left[
              \left(1+h'\gamma(t,z)\right)^{-\theta}
                - 1
         \right]\nu(dz)
\end{eqnarray}
with $f$ defined in~\eqref{eq_JDRSAM_func_f}.
\\

Under Assumption~\ref{As_Factors} (\ref{As_Factors_iii}), the functional $F$ satisfies the ellipticity property. Although $F$ is not generally monotone, $g(s,x,h)$ is bounded in $x$ (see Remark~\ref{prop_JDRSAM_g_bounded_Lipschitz}). In this case, a standard discounting technique can be used to circumvent the absence of monotonicity (see for example Proposition II.9.1 in~\cite{flso06} for the general idea and the proof of Theorem~\ref{theo_JDRSAM_comparison_parabolicPDE} below for an application in our case).

\begin{lemma}\label{lem_func_fv_Lipschitz}
        The functional $f_v$ is Lipschitz continuous in $t$ and $x$
\begin{eqnarray}
                |f_v(t,y)-f_v(s,x)| \leq C_f\left(|t-s|+|y-x|\right)
\end{eqnarray}
for some constant $C_f>0$.
\end{lemma}

\begin{proof}
        By Assumption~\ref{As_Factors} (\ref{As_Factors_i}) and Assumption~\ref{As_Securities} (\ref{As_Securities_ix})  $b(t,x)$ and $\Lambda\Sigma(t,x)'$ are Lipschitz continuous in $t$ and $x$. Moreover, by Assumption~\ref{as_JDRSAM_uncorrelatedjumps},
\begin{eqnarray}
        \int_{\mathbf{Z}}\xi(t,x,z)\left[
              \left(1+h'\gamma(t,z)\right)^{-\theta}
                - 1
         \right]\nu(dz)
        = 0
\end{eqnarray}
and the result follows.
\end{proof}


\begin{theorem}\label{theo_JDRSAM_viscositysol}
$\tilde{\Phi}$ is a Lipschitz continuous viscosity solution of the RS HJB PIDE~\eqref{eq_JDRSAM_exptrans_HJBPDE} on $[0,T] \times \mathbb{R}^n$, subject to terminal condition~\eqref{eq_JDRSAM_exptrans_HJBPDE_termcond}.
\end{theorem}

\begin{proof} We first show that $\tilde{\Phi}$ is both a, possibly discontinuous, viscosity subsolution and viscosity supersolution.
 
\noindent\textbf{Step 1: Viscosity Subsolution.}
Let $(t_0,x_0) \in Q := [0,t]\times\mathbb{R}^n$ and $u \in C^{1,2}(Q)$ satisfy
\begin{equation}
    0 = (\tilde{\Phi}^* - u)(t_0,x_0) = \max_{(x,t) \in Q} (\tilde{\Phi}^*(t,x) - u(t,x))
\end{equation}
and hence, on $Q$,
\begin{equation}\label{eq_JDRSAM_theoviscositysol_step2_ineq1}
    \tilde{\Phi} \leq \tilde{\Phi}^* \leq u.
\end{equation}
Let $(t_k,x_k)$ be a sequence in $Q$ such that
\begin{equation}
    \lim_{k \to \infty} (t_k,x_k) = (t_0,x_0)
                \nonumber
\end{equation}
\begin{equation}
    \lim_{k \to \infty} \tilde{\Phi}(t_k,x_k) = \tilde{\Phi}^*(t_0,x_0)
                \nonumber
\end{equation}
and define the sequence $\left\{\psi\right\}_k$ as $\psi_k :=
\tilde{\Phi}(t_k,x_k) - u(t_k,x_k)$. Since $u$ is of class $C^{1,2}$, $\lim_{k \to \infty} \psi_k
= 0$.
\\

Fix $h \in J$ and consider a constant control $\hat{h} = h$. Denote
by $X^k$ the state process with initial data $X_{t_k}^{k} = x_k$
and, for $k>0$, define the stopping time
\begin{equation}
    \tau_k := \inf\left\{
        s > t_k : (s-t_k,X_s^k-x_k) \notin
        [0,\delta_k)
        \times \alpha \mathscr{B}_n \right\}
                                    \nonumber
\end{equation}
for a given constant $\alpha > 0$ and where $\mathscr{B}_n$ is the
unit ball in $\mathbb{R}^n$ and
\begin{equation}
    \delta_k := \sqrt{\psi_k}\left(1-\mathit{1}_{\left\{0\right\}}(\psi_k)\right)
        + k^{-1}\mathit{1}_{\left\{0\right\}}(\psi_k)
                            \nonumber
\end{equation}
From the definition of $\tau_k$, we see that $\lim_{k \to \infty}
\tau_k = t_0$.
\\

The value function $\tilde{\Phi}$ satisfies the Dynamic Programming Principle, by Theorem 3.5 of Bouchard and Touzi \cite{BouTou11}, so that 
\begin{equation}\label{DPP}
    \tilde{\Phi}(t_k,x_k)
            \leq \mathbf{E}_{t_k,x_k} \left[ \exp \left\{ \theta \int_{t_k}^{\tau_k}
                g(s,X_s,\hat{h}_s) ds \right\} \tilde{\Phi}(\tau_k, X_{\tau_k}^{k})\right]     
\end{equation}
where $\mathbf{E}_{t_k,x_k} \left[ \cdot \right]$ represents the expectation under the measure $\mathbb{P}$ given initial data $(t_k,x_k)$. See Remark \ref{DPP-remark} for further clarification of this point.
\\

By inequality~\eqref{eq_JDRSAM_theoviscositysol_step2_ineq1},
\begin{eqnarray}
    \tilde{\Phi}(t_k,x_k)
            &\leq& \mathbf{E}_{t_k,x_k} \left[ \exp \left\{ \theta \int_{t_k}^{\tau_k}
                g(s,X_s,\hat{h}_s) ds \right\} u(\tau_k, X_{\tau_k}^{k})\right]
                                                                                \nonumber
\end{eqnarray}
and hence by definition of $\psi_k$,
\begin{eqnarray}
    u(t_k,x_k) + \psi_k
            &\leq& \mathbf{E}_{t_k,x_k} \left[ \exp \left\{ \theta \int_{t_k}^{\tau_k}
                g(s,X_s,\hat{h}_s) ds \right\} u(\tau_k, X_{\tau_k}^{k})\right]
                                                                                \nonumber
\end{eqnarray}
i.e.
\begin{eqnarray}
    \psi_k
    &\leq& \mathbf{E}_{t_k,x_k} \left[ \exp \left\{ \theta \int_{t_k}^{\tau_k}
            g(s,X_s,\hat{h}_s) ds \right\} u(\tau_k, X_{\tau_k}^{k})\right] - u(t_k,x_k)
                                                                                \nonumber
\end{eqnarray}

Define $Z(t_k) = \theta \int_{t_k}^{\tau_k}
g(s,X_s,\hat{h}_s) ds$, then
\begin{equation}
    d\left(e^{Z_s}\right) := \theta g(s,X_s,\hat{h}_s)e^{Z_s}ds
                                                                    \nonumber
\end{equation}
Also, by It\^o,
\begin{eqnarray}
    du_s
    &=& \left\{
        \frac{\partial u}{\partial s} + \mathcal{L}u \right\} ds
    +   Du'\Lambda(s) dW_{s}
                                                                    \nonumber\\
    &&
    +   \int_{\mathbf{Z}}\left\{
            u\left(s,X(s^-)+\xi(s,X(s^-),z)\right) - u\left(s,X(s^-)\right)
        \right\}\tilde{N}_{\textbf{p}}(ds,dz)
                                                                    \nonumber
\end{eqnarray}
for $s \in \left[t_k, \tau_k \right]$ and where $\mathcal{L}$ is the generator of $X(t)$ given in~\eqref{eq_JDRSAM_generator_X}.
\\

By the It\^o product rule, and since $dZ_s \cdot u_s = 0$, we get
\begin{equation}
    d\left(u_s e^{Z_s}\right) = u_s d\left(e^{Z_s}\right) + e^{Z_s} du_s
                                        \nonumber
\end{equation}
and hence for $t \in [t_k, \tau_k]$
\begin{eqnarray}
    u(t,X_t^k)e^{Z_t} &=&
        u(t_k,x_k)e^{Z_{t_k}}
        + \theta\int_{t_k}^{t} u(s,X_s^k)g(s,X_s^k,\hat{h}_s)e^{Z_s}ds
                                                    \nonumber\\
    &&  + \int_{t_k}^{t}
            \left( \frac{\partial u}{\partial s}(s,X_s^k)
            + \mathcal{L}u(s,X_s^k) e^{Z_s} \right)ds
        + \int_{t_k}^{t} Du'\Lambda(s) dW_{s}
                                                    \nonumber\\
    &&  + \int_{t_k}^{t}\int_{\mathbf{Z}}\left\{
            u\left(s,X^k(s^-)+\xi(s,X^k(s^-),z)\right)
                                - u\left(s,X^k(s^-)\right)
        \right\}\tilde{N}_{\textbf{p}}(dt,dz)
                                                                    \nonumber
\end{eqnarray}

Noting that $u(t_k,x_k)e^{Z_{t_k}} = u(t_k,x_k)$ and taking the expectation with respect to the initial data $(t_k,x_k)$, we get
\begin{eqnarray}
    && \mathbf{E}_{t_k,x_k} \left[ u(t,X_t)e^{Z_t} \right]
                                                    \nonumber\\
    &=&  u(t_k,x_k)e^{Z_{t_k}}
        + \mathbf{E}_{t_k,x_k} \left[
          \int_{t_k}^{t}
                \left(\frac{\partial u}{\partial s}(s,X_s) + \mathcal{L}u(s,X_s)
                +\theta u(s,X_s)g(s,X_s,\hat{h}_s)\right)e^{Z_s}ds
          \right]
                                                    \nonumber
\end{eqnarray}

In particular, for $t = \tau_k$,
\begin{eqnarray}
    \psi_k
    &\leq& \mathbf{E}_{t_k,x_k} \left[ u(\tau_k,X_{\tau_k})e^{Z_{\tau_k}} \right]
             - u(t_k,x_k)e^{Z_{t_k}}
                                                    \nonumber\\
    &=&  \mathbf{E}_{t_k,x_k} \left[
          \int_{t_k}^{\tau_k}
                \left(\frac{\partial u}{\partial s}(s,X_s) + \mathcal{L}u(s,X_s)
                +\theta u(s,X_s)g(s,X_s,\hat{h}_s)\right)e^{Z_s}ds
          \right]
                                                    \nonumber
\end{eqnarray}
and thus
\begin{eqnarray}
    \frac{\psi_k}{\delta_k}
    &\leq& \frac{1}{\delta_k}\left(
             \mathbf{E}_{t_k,x_k,} \left[ u(\tau_k,X_{\tau_k})e^{Z_{\tau_k}} \right]
             - u(t_k,x_k)e^{Z_{t_k}}\right)
                                                    \nonumber\\
    &=&  \frac{1}{\delta_k}\left(
          \mathbf{E}_{t_k,x_k} \left[
          \int_{t_k}^{\tau_k}
                \left(\frac{\partial u}{\partial s}(s,X_s) + \mathcal{L}u(s,X_s)
                +\theta u(s,X_s)g(s,X_s,\hat{h}_s)\right)e^{Z_s}ds
          \right] \right)
                                                    \nonumber
\end{eqnarray}

As $k \to \infty$, $t_k \to t_0$, $\tau_k \to t_0$, $\frac{\psi_k}{\delta_k} \to 0$ and
\begin{eqnarray}
    && \frac{1}{\delta_k}\left(
          \mathbf{E}_{t_k,x_k} \left[
          \int_{t_k}^{t}
                \left(\frac{\partial u}{\partial s}(s,X_s) + \mathcal{L}u(s,X_s)
                +\theta u(s,X_s)g(s,X_s,\hat{h}_s)\right)e^{Z_s}ds
          \right] \right)
                                                    \nonumber\\
    &\to&
        \frac{\partial u}{\partial s}(s,X_s) + \mathcal{L}u(s,X_s)
        +\theta u(s,X_s)g(s,X_s,\hat{h}_s)
                                                    \nonumber
\end{eqnarray}
a.s. by the Bounded Convergence Theorem, since the random variable
\begin{equation}
    \frac{1}{\delta_k} \int_{t_k}^{t} \left(
        \frac{\partial u}{\partial s}(s,X_s)
        + \mathcal{L}u(s,X_s)
        +\theta u(s,X_s)g(s,X_s,\hat{h}_s)
    \right)e^{Z_s}ds
                                                    \nonumber
\end{equation}
is bounded for large enough $k$.
\\

Hence, we conclude that since $\hat{h}_s$ is arbitrary,
\begin{equation}
    \frac{\partial u}{\partial s}(s,X_s) + \mathcal{L}u(s,X_s)
    +\theta u(s,X_s)g(s,X_s,\hat{h}_s)
    \geq 0
                                                    \nonumber
\end{equation}
i.e.
\begin{equation}
    -\frac{\partial u}{\partial s}(s,X_s) - \mathcal{L}u(s,X_s)
    -\theta u(s,X_s)g(s,X_s,\hat{h}_s)
    \leq 0
                                                    \nonumber
\end{equation}

This argument proves that $V$ is a (discontinuous) viscosity subsolution of the PDE~\eqref{eq_JDRSAM_exptrans_HJBPDE} on $[0,t) \times \mathbb{R}^n$ subject to terminal condition $\tilde{\Phi}(T, x) = e^{\theta g_T(x;T)}$.
\\

\noindent\textbf{Step 2: Viscosity Supersolution}\\

This step in the proof is a slight adaptation of the proof for classical control problems in Touzi~\cite{to02}. Let $(t_0,x_0) \in Q$ and $u \in C^{1,2}(Q)$ satisfy
\begin{equation}\label{eq_JDRSAM_theoviscositysol_step3_ineq1}
    0 = (\tilde{\Phi}_* - u)(t_0,x_0) < (\tilde{\Phi}_* - u)(t,x) \textrm{ for } Q\backslash{(t_0,x_0)}
\end{equation}

We intend to prove that at $(t_0,x_0)$
\begin{equation}
   \frac{\partial u}{\partial t}(t,x)
    + \inf_{h \in \mathcal{H}}\left\{
        \mathcal{L}^{h} u(t,x) - \theta g(t,x,h)
    \right\}
    \leq 0
                            \nonumber
\end{equation}
by contradiction. Thus, assume that
\begin{equation}\label{eq_JDRSAM_theoviscositysol_step3_contradictionPDE}
   \frac{\partial u}{\partial t}(t,x)
    + \inf_{h \in \mathcal{H}}\left\{
        \mathcal{L}^{h} u(t,x) - \theta g(t,x,h)
    \right\}
    > 0
\end{equation}
at $(t_0,x_0)$.

Since $\mathcal{L}^h u$ is continuous, there exists an open neighbourhood $\mathcal{N}_{\delta}$ of $(t_0,x_0)$ defined for $\delta > 0$ as
\begin{equation}
    \mathcal{N}_{\delta} := \left\{(t,x): (t-t_0,x-x_0) \in (-\delta,\delta) \times \delta \mathscr{B}_n, \textrm{ and~\eqref{eq_JDRSAM_theoviscositysol_step3_contradictionPDE} holds}\right\}
\end{equation}
\\

Note that by~\eqref{eq_JDRSAM_theoviscositysol_step3_ineq1} and since $\tilde{\Phi} > \tilde{\Phi}_{*} > u $,
\begin{equation}
    \min_{Q \backslash \mathcal{N}_{\delta}}\left(\tilde{\Phi} - u \right) > 0
                \nonumber
\end{equation}
\\

For $\rho >0$, consider the set $J^{\rho}$ of $\rho$-optimal controls $h^{\rho}$ satisfying
\begin{equation}\label{eq_JDRSAM_visc_HJBderivation_step3_def_epsilonoptimal}
    \tilde{I}(t_0,x_0,h^{\rho}) \leq \tilde{\Phi}(t_0,x_0) + \rho
\end{equation}
\\

Also, let $\epsilon > 0$, $\epsilon \leq \gamma$ be such that
\begin{equation}\label{eq_JDRSAM_visc_HJBderivation_step3_epsilon}
   \min_{Q \backslash \mathcal{N}_{\delta}}
        \left(\tilde{\Phi} - u \right)
        \geq 3\epsilon e^{-\delta \theta M_{\delta}}
    > 0
\end{equation}
where $M_{\delta}$ is defined as
\begin{equation}
    M_{\delta} := \max_{(t,x) \in \mathcal{N}_{\delta}^J, h \in J^{\rho}}\left(-g(x,h),0\right)
                                        \nonumber
\end{equation}
for
\begin{equation}
    \mathcal{N}_{\delta}^J := \left\{(t,x): (t-t_0,x-x_0) \in (-\delta,\delta) \times (\zeta+\delta) \mathscr{B}_n \right\}
\end{equation}
and
\begin{equation}
    \zeta := \max_{z \in \mathbb{Z}} \| \xi(z)\|
                                        \nonumber
\end{equation}
Note that $M_{\delta} < \infty$ by boundedness of $g$ (see Property~\ref{prop_JDRSAM_g_bounded_Lipschitz}).
\\

Now let $(t_k,x_k)$ be a sequence in $\mathcal{N}_{\delta}$ such that
\begin{equation}
    \lim_{k \to \infty} (t_k,x_k) = (t_0,x_0)
                \nonumber
\end{equation}
and
\begin{equation}
    \lim_{k \to \infty} \tilde{\Phi}(t_k,x_k) = \tilde{\Phi}_*(t_0,x_0)
                \nonumber
\end{equation}
Since $(\tilde{\Phi}-u)(t_k,x_k) \to 0$, we can assume that the sequence $(t_k,x_k)$ satisfies
\begin{equation}\label{eq_JDRSAM_theoviscositysol_step3_ineq2}
    \lvert (\tilde{\Phi}-u)(t_k,x_k) \rvert
    \leq  \epsilon,
    \qquad \textrm{for } k \geq 1
\end{equation}
for $\epsilon$ defined by~\eqref{eq_JDRSAM_visc_HJBderivation_step3_epsilon}
\\

Consider the $\epsilon$-optimal control $h_k^{\epsilon}$, denote by $\tilde{X}_k^\epsilon$ the controlled process defined by the control process $h_k^{\epsilon}$ and introduce the stopping time
\begin{equation}
    \tau_k := \inf\left\{s>\tau_k : (s,\tilde{X}_k^\epsilon(s)) \notin \mathcal{N}_{\delta} \right\}
                            \nonumber
\end{equation}
Note that since we assumed that $-\infty \leq \xi_{i}^{\textrm{min}} \leq \xi_{i} \leq \xi_{i}^{\textrm{max}} < \infty$ for $i = 1, \ldots, n$ and since $\nu$ is assumed to be bounded then $X(\tau)$ is also finite and in particular,
\begin{equation}\label{eq_JDRSAM_theoviscositysol_step3_ineq3}
    (\tilde{\Phi}-u)(\tau_k,\tilde{X}_k^\epsilon(\tau_k))
    \geq (\tilde{\Phi}_{*} - u)(\tau_k,\tilde{X}_k^\epsilon(\tau_k))
    \geq 3\epsilon e^{-\delta \theta M_{\delta}}
\end{equation}
\\

Choose $\mathcal{N}_{\delta}^{J}$ so that $(\tau,\tilde{X}^\epsilon(\tau)) \in \mathcal{N}_{\delta}^{J}$. In particular, since $X^\epsilon(\tau)$ is finite then $\mathcal{N}_{\delta}^{J}$  can be defined to be a strict subset of $Q$ and we can effectively use the local boundedness of $g$ to establish $M_{\delta}$.
\\

Let $Z(t_k) = \theta \int_{t_k}^{\bar{\tau}_k} g(s,\tilde{X}_s^{\epsilon},h_s^{\epsilon}) ds$, since $\tilde{\Phi} \geq \tilde{\Phi}_{*}$ and by~\eqref{eq_JDRSAM_theoviscositysol_step3_ineq2} and~\eqref{eq_JDRSAM_theoviscositysol_step3_ineq3},
\begin{eqnarray}
    &&  \tilde{\Phi}(\tau_k,\tilde{X}_k^\epsilon(\tau_k))e^{Z(\tau_k)}
        -\tilde{\Phi}(t_k,x_k)e^{Z(t_k)}
                                        \nonumber\\
    &\geq&
        u(\tau_k,\tilde{X}_k^\epsilon(\tau_k))e^{Z(\tau_k)}
        -\tilde{\Phi}(t_k,x_k)e^{Z(t_k)}
        + 3\epsilon e^{-\delta \theta M_{\delta}}e^{Z(\tau_k)}
        - \epsilon
                                \nonumber\\
    &\geq&
        \int_{t_k}^{\tau_k} d\left(u(s,\tilde{X}_k^\epsilon(s))e^{Z_s}\right)
        + 2\epsilon
                                \nonumber
\end{eqnarray}
i.e.
\begin{eqnarray}
   \tilde{\Phi}(t_k,x_k)
    &\leq&
        \tilde{\Phi}(\tau_k,\tilde{X}_k^\epsilon(\tau_k))e^{Z(\tau_k)}
        - \int_{t_k}^{\tau_k} d\left(u(s,\tilde{X}_k^\epsilon(s))e^{Z_s}\right)
        - 2\epsilon
                            \nonumber
\end{eqnarray}

Taking expectation with respect to the initial data $(t_k,x_k)$,
\begin{eqnarray}
   \tilde{\Phi}(t_k,x_k)
    &\leq&
        \mathbf{E}_{t_k,x_k}\left[
            \tilde{\Phi}(\tau_k,\tilde{X}_k^\epsilon(\tau_k))e^{Z(\tau_k)}
            - \int_{t_k}^{\tau_k} d\left(u(s,\tilde{X}_k^\epsilon(s))e^{Z_s}\right)
        \right]
        - 2\epsilon
                            \nonumber
\end{eqnarray}

Note that by the It\^o product rule,
\begin{eqnarray}
    &&   d\left(u(s,\tilde{X}_k^\epsilon(s))e^{Z_s}\right)
                                    \nonumber\\
    &=& u_s d\left(e^{Z_s}\right) + e^{Z_s} du_s
                                    \nonumber\\
    &=& \frac{\partial u}{\partial t}(t,x)+ \mathcal{L}^{h} u(t,x)
            + \theta g(t,x,h)
                                    \nonumber
\end{eqnarray}
Since we assumed that
\begin{equation}
   -\frac{\partial u}{\partial t}(t,x)- \mathcal{L}^{h} u(t,x)
    -\theta g(t,x,h) < 0
                            \nonumber
\end{equation}
then
\begin{equation}
    -\int_{t_k}^{\tau_k} d\left(u(s,\tilde{X}_k^\epsilon(s))e^{z_s}\right)
     < 0
                            \nonumber
\end{equation}
and therefore
\begin{eqnarray}
    \tilde{\Phi}(t_k,x_k)
    &\leq&
        \mathbf{E}_{t_k,x_k}\left[
            \tilde{\Phi}(\tau_k,\tilde{X}_k^\epsilon(\tau_k))e^{Z(\tau_k)}
            - \int_{t_k}^{\tau_k} d\left(u(s,\tilde{X}_k^\epsilon(s))e^{Z_s}\right)
        \right]
        - 2\epsilon
                            \nonumber\\
    &\leq&  -2\epsilon + \mathbf{E}
        \left[ \exp \left\{ \theta \int_{t_k}^{\tau_k} g(X_s,h_k^{\epsilon}(s))
        ds\right\}\tilde{\Phi}(\tau_k,\tilde{X}_k^\epsilon(\tau_k))\right]
                            \nonumber\\
    &\leq&  -2\epsilon + \tilde{I}(t_k,x_k,h_k^{\epsilon})
                            \nonumber\\
    &\leq&  \tilde{\Phi}(t_k,x_k) - \epsilon
                            \nonumber
\end{eqnarray}
where the third inequality follows from the Dynamic Programming Principle and the last inequality follows from the definition of $\epsilon$-optimal controls\footnote{For $\epsilon >0$, the set $J^{\epsilon}$ of $\epsilon$-optimal controls $h^{\epsilon}$ is the set of control satisfying
\begin{equation}
    \tilde{I}(t_0,x_0,h^{\epsilon}) \leq \tilde{\Phi}(t_0,x_0) + \epsilon
                                                                                                                                \nonumber
\end{equation}}.
\\

Hence, equation~\eqref{eq_JDRSAM_theoviscositysol_step3_contradictionPDE},
\begin{equation}
   \frac{\partial u}{\partial t}(t,x)
    + \inf_{h \in \mathcal{H}}\left\{
        \mathcal{L}^{h} u(t,x) - \theta g(t,x,h)
    \right\}
    > 0
                                            \nonumber
\end{equation}
is false and we have shown that
\begin{equation}
   \frac{\partial u}{\partial t}(t,x)
    + \inf_{h \in \mathcal{H}}\left\{
        \mathcal{L}^{h} u(t,x) - \theta g(t,x,h)
    \right\}
    \leq 0
                                            \nonumber
\end{equation}
\\

This argument therefore proves that $\tilde{\Phi}$ is a (discontinuous) viscosity supersolution of the PDE~\eqref{eq_JDRSAM_exptrans_HJBPDE} on $[0,t) \times \mathbb{R}^n$ subject to terminal condition $\tilde{\Phi}(T, x) = e^{\theta g_T(x;T)}$. 
\\

\noindent\textbf{Step 3: Viscosity Solution.} Since $\tilde{\Phi}$ is both a (discontinuous) viscosity subsolution and a supersolution of~\eqref{eq_JDRSAM_exptrans_HJBPDE}, it is a (discontinuous) viscosity solution. But we already know from Proposition \ref{prop_tildePhi_Lipschitz} that $\tilde{\Phi}$ is Lipschitz continuous.
\hfill\end{proof}
\\

\begin{remark}\label{DPP-remark}
        Theorem 3.5 of Bouchard and Touzi~\cite{BouTou11} provides the Dynamic Programming Principle (DPP) \eqref{DPP} we needed above to show the existence of a viscosity solution. They state it in `weak' form, meaning that on the right-hand side the value function is replaced by its upper semi-continuous envelope. However, in our case we know \emph{a priori} that the value function is continuous, so it is valid to state the DPP in its classic form as in \eqref{DPP}. We can check that Conditions A1--A4 of the theorem are valid in our `weak solutions' formulation of the control problem.

In Section 5 of \cite{BouTou11}, Bouchard and Touzi directly address the DPP question and characterization of optimality by viscosity solutions of the HJB equation for a class of controlled Markov jump-diffusions. We cannot use their results (Proposition 5.4 and Corollary 5.6) directly as the formulation is subtly different and the conditions exclude infinite-activity jumps, although the core of their argument does not depend on the nature of the jumps. As stated in Remark 5.1 of \cite{BouTou11} their assumption is necessary to match the strictest set of assumptions required by Barles and Imbert~\cite{baim08}.  
It is widely appreciated that the heart of viscosity solution theory lies in the uniqueness theorems and, as seen for example in  Barles and Imbert~\cite{baim08}, to prove the necessary comparison theorems stronger conditions are generally required than those needed for existence. Our strategy is to by-pass this question entirely by taking a route that only requires uniqueness of viscosity solutions for PDEs---where a large literature exists---rather than PIDEs where results are sparser.
\end{remark}
\\

\begin{corollary}\label{theo_JDRSAM_viscositysol_phi}
$\Phi$ is a continuous viscosity solution of the RS HJB PIDE~\eqref{eq_JDRSAM_HJBPDE} on $[0,T] \times \mathbb{R}^n$, subject to terminal condition~\eqref{eq_JDRSAM_HJBPDE_termcond}.
\end{corollary}

\begin{proof} This follows from the change of variable property (see for example Proposition 2.2 in Touzi~\cite{to02}) which applies here because 
(a) the value function $\tilde{\Phi}$ is bounded by Proposition~\ref{prop_JDRSAM_tildePhi_bounded}, and (b) the function $\varphi(x) = e^{-\theta x}$ is of class $C_1(\mathbb{R})$ and $\frac{d\varphi}{dx} < 0$. \hfill\end{proof}

\subsection{From PIDE to PDE}
We define the functional $f_{a}$, the non-local function $\check{\mathcal{I}}$ and the functional $d_a$ respectively as:
\begin{eqnarray}\label{def_JDRSAM_function_tilde_f}
        f_{a}(x,h)
        &:=&    f(x,h)
        - \int_{\mathbf{Z}} \xi(t,x,z) \nu(dz)
                                                                                                                                                \nonumber\\
        &=&     b(t,x)
                -       \theta\Lambda\Sigma(t,x)'h(s)
        +       \int_{\mathbf{Z}}\xi(t,x,z)\left[
              \left(1+h'\gamma(t,z)\right)^{-\theta}
               -        \mathit{1}_{\mathbf{Z}_0}(z)
                                        -       1
         \right]\nu(dz)
                                                                                                                                                \nonumber\\
\end{eqnarray}
where $f$ is defined in~\eqref{eq_JDRSAM_func_f},
\begin{eqnarray}\label{def_JDRSAM_function_tilde_d}
        d_{a}^{\tilde{\Phi}}(t,x)
        = \check{\mathcal{I}}\left[t,x,\tilde{\Phi}(t,x)\right]
        =       \int_{\mathbf{Z}} \left\{
              \tilde{\Phi}\left(t,x+\xi(t,x,z)\right)
            - \tilde{\Phi}(t,x)
        \right\} \nu(dz)
\end{eqnarray}

\begin{remark}
Under Assumption~\ref{as_JDRSAM_uncorrelatedjumps},
\begin{eqnarray}
        \int_{\mathbf{Z}}\xi(t,x,z)\left[
       \left(1+h'\gamma(t,z)\right)^{-\theta}
        -       \mathit{1}_{\mathbf{Z}_0}(z)
                        -       1
    \right]\nu(dz)
        =
        - \int_{\mathbf{Z}_0}\xi(t,x,z)\nu(dz)
\end{eqnarray}
and therefore
\begin{eqnarray}
        f_{a}(x,h)
        &=&     b(t,x)
                -       \theta\Lambda\Sigma(t,x)'h(s)
        -       \int_{\mathbf{Z}_0} \xi(t,x,z)\nu(dz)
\end{eqnarray}
\\
\end{remark}

Under Assumption~\eqref{as_factorjumps_xi_integrable}, $f_a$ is well defined. With this notation, we can express the risk-sensitive integro-differential HJB PIDE in terms of an equivalent parabolic PDE \eqref{eq_JDRSAM_exptrans_HJBPDE_parabolic_outline} as stated in Step 2 of our outline proof above.

The functions $f_a$ and $\check{\mathcal{I}}$ have the following properties:
 
\begin{lemma}\label{lem_JDRSAM_tildef_Lipshcitz_cond2}
      The function $f_a$ is Lipschitz continuous and bounded.
\end{lemma}

\begin{proof}
        This follows from the continuity and boundedness of $f_a$ and from the boundedness of $\int_{\mathbf{Z}_0} \xi(t,x,z) \nu(dz)$, using Assumptions 1(v) and 1(vi).
\end{proof}

\begin{lemma}\label{lem_JDRSAM_d_bounded_cond2}
     The function $\check{\mathcal{I}}\left[t,x,\tilde{\Phi}(t,x)\right]$ is continuous.
\end{lemma}

\begin{proof}
The proof of continuity follows along the lines of Lemma 3.2 in Davis, Guo and Wu~\cite{guwu09}, and relies on the Lipschitz continuity of $\tilde{\Phi}$ (Proposition~\ref{prop_tildePhi_Lipschitz}) and on Assumption 1(vi).

\end{proof}

\subsection{Viscosity Solution to the PDE~\eqref{eq_JDRSAM_exptrans_HJBPDE_parabolic_outline}}

Using the notation introduced in the previous step, we can express the risk-sensitive integro-differential HJB PDE as:
\begin{eqnarray}\label{eq_JDRSAM_exptrans_HJBPDE_viscosity2}
        \frac{\partial \tilde{\Phi}}{\partial t}(t,x)
                + \frac{1}{2} \textrm{tr} \left( \Lambda\Lambda'(t,x) D^2 \tilde{\Phi}(t,x)\right)
                        + H_{a}(t,x,\tilde{\Phi},D\tilde{\Phi})
                        + \check{\mathcal{I}}\tilde{\Phi}(t,x)
    &=& 0
\end{eqnarray}
subject to terminal condition $\tilde{\Phi}(T, x) = 1$.
\newline

\begin{remark}\label{rk_JDRSAM_Ha_Lipschitz}
By Remark~\ref{rk_JDRSAM_H_Lipschitz}, the function $H_a$ satisfies a Lipschitz condition as well as the linear growth condition
\begin{eqnarray}
        \lvert H_{a}(s,x,r,p) \rvert
        \leq C_a\left(1 + \lvert p \rvert\right)
        ,       \quad   \forall (s,x) \in Q_0
                                                        \nonumber
\end{eqnarray}
\end{remark}

We now present an alternative definition first suggested by Pham~\cite{ph98} and then formalized in the context of impulse control by Davis, Guo and Wu~\cite{guwu09}. In this definition, the integro-differential operator is evaluated using the actual solution.
\newline

\begin{definition}[Viscosity Solution (Test Functions in the Local Terms Only)]\label{def_JDRSAM_viscositysol_testfunc_alt}
A bounded function $u \in USC([0,T]\times\mathbb{R}^n)$ is a \emph{viscosity subsolution} of~\eqref{eq_JDRSAM_exptrans_HJBPDE_viscosity}, if for any bounded test function $\psi \in C^{1,2}([0,T]\times\mathbb{R}^n)$, if $(x,t)$ is a global maximum point of $u-\psi$, then
\begin{eqnarray}
        - \frac{\partial \psi}{\partial t}
   - \frac{1}{2} \textrm{tr} \left( \Lambda\Lambda'(t,x) D^2\psi(t,x)\right)
        - H_{a}(t,x,u(t,x),D\psi(t,x))
        - \check{\mathcal{I}}u(t,x)
        \leq 0
\end{eqnarray}

A bounded function $u \in LSC([0,T]\times\mathbb{R}^n)$ is a \emph{viscosity supersolution} of~\eqref{eq_JDRSAM_exptrans_HJBPDE_viscosity}, if for any bounded test function $\psi \in C^{1,2}([0,T]\times\mathbb{R}^n)$, if $(x,t)$ is a global minimum point of $u-\psi$, then
\begin{eqnarray}
        - \frac{\partial \psi}{\partial t}
   - \frac{1}{2} \textrm{tr} \left( \Lambda\Lambda'(t,x) D^2\psi(t,x)\right)
        - H_{a}(t,x,u(t,x),D\psi(t,x))
        - \check{\mathcal{I}}u(t,x)
        \geq 0
\end{eqnarray}

A bounded function $u$ whose upper semicontinuous and lower semicontinuous envelopes are a viscosity subsolution and a viscosity supersolution of~\eqref{eq_JDRSAM_HJBPDE} is a viscosity solution of~\eqref{eq_JDRSAM_HJBPDE}.
\end{definition}

\begin{proposition}\label{prop_viscosity_equivalence}
        The definitions of viscosity solutions~\ref{def_JDRSAM_viscositysol_testfunc} and~\ref{def_JDRSAM_viscositysol_testfunc_alt} are equivalent.
\end{proposition}

\begin{proof}
        The argument follows from Davis, Guo and Wu~\cite{guwu09}.
\end{proof}
\newline

Definition~\ref{def_JDRSAM_viscositysol_testfunc_alt} and Proposition~\ref{prop_viscosity_equivalence} enable us to relate the viscosity solutions of HJB PIDE~\eqref{eq_JDRSAM_exptrans_HJBPDE} and HJB PDE~\eqref{eq_JDRSAM_exptrans_HJBPDE_parabolic_outline}.
\newline

\begin{proposition}\label{prop_PIDE_to_PDE}
        The function $\tilde{\Phi}$ is a viscosity solution of~\eqref{eq_JDRSAM_exptrans_HJBPDE_parabolic_outline}. 
\end{proposition}

\begin{proof}
        The proposition is a direct consequence of Definition~\ref{def_JDRSAM_viscositysol_testfunc_alt} and Proposition~\ref{prop_viscosity_equivalence}.
\end{proof}
\newline

We now need to establish uniqueness for the PDE \eqref{eq_JDRSAM_exptrans_HJBPDE_parabolic_outline}.

\subsection{Uniqueness of the Viscosity Solution to the PDE~\eqref{eq_JDRSAM_exptrans_HJBPDE_parabolic_outline}}
The argument here follows the conventional steps of the `Users' Guide' \cite{crisli92}.
\begin{theorem}[Comparison Result for the Parabolic PDE]\label{theo_JDRSAM_comparison_parabolicPDE}
        Under Assumption~\eqref{as_factorjumps_xi_integrable}, if $u$ is a bounded usc subsolution of~\eqref{eq_JDRSAM_exptrans_HJBPDE_parabolic_outline} subject to terminal condition~\eqref{eq_JDRSAM_exptrans_HJBPDE_termcond} and $v$ is a bounded lsc subsolution of~\eqref{eq_JDRSAM_exptrans_HJBPDE_parabolic_outline} subject to terminal condition~\eqref{eq_JDRSAM_exptrans_HJBPDE_termcond}, then $u \leq v$ on $[0,T]\times\mathbb{R}^n$.
\end{theorem}

\begin{proof} Let $u \in USC([0,T]\times\mathbb{R}^{n})$ be a viscosity subsolution of~\eqref{eq_JDRSAM_exptrans_HJBPDE_viscosity} and $v \in LSC([0,T]\times\mathbb{R}^{n})$ be a viscosity supersolution of~\eqref{eq_JDRSAM_exptrans_HJBPDE_viscosity}. As is usual in the derivation of comparison results, we argue by contradiction and assume that
\begin{eqnarray}\label{as_JDRSAM_comparison_theo_contradiction_parabolic}
    \sup_{[0,T] \times\mathbb{R}^n} \left[ u(t,x) - v(t,x) \right] >0
\end{eqnarray}
Define $g_0$ as
\begin{equation}\label{eq_JDRSAM_g0_def_parabolic}
        g_0
        = \theta\inf_{(t,x,h) \in (0,T)\times\mathbb{R}^n\times\mathcal{U}}
                        \left\{ -g(t,x,h) \right\} \vee 0
\end{equation}
Because $g$ is strictly convex in $h$ and it bounded in $(t,x)$ for a given $h \in \mathcal{U}$, $g_0$ is well defined and in particular $0 \leq g_0 < \infty$. Let
\begin{eqnarray}
        \bar{u}(t,x) := e^{-g_0t}u(t,x)
                                                                                        \nonumber\\
        \bar{v}(t,x) := e^{-g_0t}v(t,x).
                                                                                        \nonumber
\end{eqnarray}
Then $\bar{u}\in USC([0,T]\times\mathbb{R}^{n})$ and $\bar{v}\in LSC([0,T]\times\mathbb{R}^{n})$ are respectively a viscosity subsolution and a viscosity supersolution of the HJB PIDE
then $\phi$ satisfies the HJB PIDE
\begin{eqnarray}\label{eq_JDRSAM_exptrans_HJBPDE_viscosity_bd_parabolic}
   - \frac{\partial \phi}{\partial t}(t,x)
        + F_b(t,x,\phi,D\phi,D^2\phi)
   &=& 0
\end{eqnarray}
subject to terminal condition~\eqref{eq_JDRSAM_exptrans_HJBPDE_termcond}, where the functional $F$ is defined as
\begin{eqnarray}\label{eq_JDRSAM_exptrans_functional_F1}
        F_b(t,x,\phi,D\phi,D^2\phi)
   &=&
        H_b(t,x,\phi,D\phi)
        -       \frac{1}{2} \textrm{tr} \left( \Lambda\Lambda'(t,x) D^2 \phi(t,x)\right)
                                                                                                        \nonumber
\end{eqnarray}
and where $H_b$ is defined as
\begin{eqnarray}\label{eq_JDRSAM_logtrans_H1_function}
   H_b(s,x,r,p)
        &=& g_0r + \sup_{h \in \mathcal{A}} \left\{
                -       f_{v}(s,x,h)'p
                        -       \theta g(s,x,h) r
        \right\}
                                                                                                                \nonumber
\end{eqnarray}
for $r \in \mathbb{R}$, $p \in \mathbb{R}^n$. Observe that
\begin{eqnarray}
        H_b(s,x,r+d,p) - H_b(s,x,r,p) \geq  g_0 d
                                                                                                                \nonumber
\end{eqnarray}
for any $r>0$ and therefore
\begin{eqnarray}\label{eq_JDRSAM_viscosity_g0_ineq_parabolic}
        F_b(s,x,r+d,p) - F_b(s,x,r,p) \geq  g_0 d.
\end{eqnarray}
The contradiction~\eqref{as_JDRSAM_comparison_theo_contradiction_parabolic} is equivalent to:
\begin{eqnarray}\label{as_JDRSAM_comparison_theo_contradiction_bar_parabolic}
    \sup_{[0,T] \times\mathbb{R}^n} \left[ \bar{u}(t,x) - \bar{v}(t,x) \right] >0
\end{eqnarray}
Now, we can apply Theorem 8.2 in Crandall, Ishii and Lions~\cite{crisli92} with Lemma~\ref{lem_JDRSAM_ModulusofContinuity_parabolic} below providing the modulus of continuity.

\end{proof}

\begin{lemma}[Modulus of Continuity]\label{lem_JDRSAM_ModulusofContinuity_parabolic}
        There exists a continuous function $\omega: \mathbb{R}^+ \to \mathbb{R}^+$ satisfying $\omega(0)=0$ and such that
\begin{eqnarray}
        \Bigg| F_b\left(t,y,v(t,x),\frac{1}{\epsilon}(x-y),B\right)
        - F_b\left(t,x,v(t,x),\frac{1}{\epsilon}(x-y),A\right) \Bigg|
        \leq \omega\left(\frac{1}{\epsilon}|x-y|^2 + |x-y|\right)
                                                                                                                                        \nonumber
\end{eqnarray}
for $\epsilon>0, t \in [0,T]$ and $x,y \in \mathbb{R}^n$ and where $A,B$ are symmetric matrices.
\end{lemma}

\begin{proof}
This proof follows closely the argument given in Lemma V.7.1 in Fleming and Soner\cite{flso06}.
\\

We have
\begin{eqnarray}
        &&                      F_b\left(t,y,\bar{v}(t,y),\frac{1}{\epsilon}(x-y),B\right)
        -  F_b\left(t,x,\bar{v}(t,y),\frac{1}{\epsilon}(x-y), A\right)
                                                                                                                        \nonumber\\
        &=&     H_b\left(t,y,\bar{v}(t,y),\frac{1}{\epsilon}(x-y)\right)
                -       H_b\left(t,x,\bar{v}(t,y),\frac{1}{\epsilon}(x-y)\right)
                                                                                                                        \nonumber\\
        &\leq&
                                \sup_{h \in \mathcal{A}} \left\{
                                        \frac{1}{\epsilon}\big|f_{a}(t,x,h) - f_{a}(t,y,h)\big|
                                                 |y-x|
                                \right\}
                                                                                                                        \nonumber\\
        &&              +       \sup_{h \in \mathcal{A}} \left\{
                                        \theta \big| [g(t,x,h) + g(t,y,h)] \bar{v}(t,y) \big|
                                \right\}
                                                                                                                        \nonumber\\
        &&
                        +       \sup_{h \in \mathcal{A}} \left\{\frac{1}{2}\textrm{tr}\left(
                                        \Lambda\Lambda'(t,x,h)A
                                -       \Lambda\Lambda'(t,y,h)B
                        \right)\right\}
                                                                                                                        \nonumber
\end{eqnarray}
For a fixed control $h$, we have by Lemma~\ref{lem_JDRSAM_tildef_Lipshcitz_cond2}
\begin{eqnarray}
        &&      \frac{1}{\epsilon}\big|f_{a}(t,x,h)
                - f_{a}(t,y,h)\big|
                                                 |y-x|
        +       \theta \big| [g(t,x,h) - g(t,y,h)] \bar{v}(t,y) \big|
                                                                                                                        \nonumber\\
        &\leq&  C_f\frac{1}{\epsilon}|y-x|^2
                                + \theta C_g |y-x|\bar{v}(t,y)
                                                                                                                        \nonumber\\
        &\leq&  C_f\frac{1}{\epsilon}|y-x|^2
                                + C_0 |y-x|
                                                                                                                        \nonumber\\
        &\leq&  C_1\left( \frac{1}{\epsilon}|y-x|^2
                                + |y-x|\right)
                                                                                                                        \nonumber
\end{eqnarray}
where $C_0 = \theta C_g sup_{(s,y)\in(0,T)\times\mathbb{R}^n}\bar{v}(s,y)$ and $C_1 = C_f \vee C_g$.
\\

Applying Theorem 8.3 in Crandall, Ishii and Lions~\cite{crisli92} for a fixed control $h$,
\begin{eqnarray}
        &&
        \textrm{tr}\left(
                \Lambda\Lambda'(t,x,h)A
                -       \Lambda\Lambda'(t,y,h)B
                \right)
                                                                                                \nonumber\\
        &=&\textrm{tr}\left(
                \left[\begin{array}{cc}
                        \Lambda\Lambda'(t,x,h)
                                &       \Lambda(t,x,h)\Lambda'(t,y,h)           \\
                        \Lambda(t,y,h)\Lambda'(t,x,h)
                                &       \Lambda\Lambda'(t,y,h)
                \end{array}\right]
                \left[\begin{array}{cc}
                        A       &       0               \\
                        0       &       -B
                \end{array}\right]
                \right)
                                                                                                \nonumber\\
        &\leq&  \frac{1}{\epsilon}
\textrm{tr}\left(
                \left[\begin{array}{cc}
                        \Lambda\Lambda'(t,x,h)
                                &       \Lambda(t,x,h)\Lambda'(t,y,h)           \\
                        \Lambda(t,y,h)\Lambda'(t,x,h)
                                &       \Lambda\Lambda'(t,y,h)
                \end{array}\right]
                \left[\begin{array}{cc}
                        I       &       -I              \\
                        -I      &       I
                \end{array}\right]
                \right)
                                                                                                \nonumber\\
        &\leq&  \frac{1}{\epsilon} \textrm{tr}\left(
                                                \Lambda\Lambda'(t,x,h)
                                        -       \Lambda(t,x,h)\Lambda'(t,y,h)
                                        -       \Lambda(t,y,h)\Lambda'(t,x,h)
                                        +       \Lambda\Lambda'(t,y,h)
                \right)
                                                                                                \nonumber\\
        &\leq&  \frac{1}{\epsilon} \textrm{tr}\left(
                                        |\Lambda(t,x,h) - \Lambda'(t,y,h)|
                                        |\Lambda(t,x,h) - \Lambda'(t,y,h)|
                \right)
                                                                                                \nonumber\\
        &\leq&  \frac{1}{\epsilon} \textrm{tr}\left(
                                        \vert|\Lambda(t,x,h)
                                                - \Lambda'(t,y,h)\vert|^2
                \right)
                                                                                                \nonumber\\
        &\leq&          \frac{C_{\Lambda}}{\epsilon} |x-y|^2
\end{eqnarray}
where the last line follows from Assumption~\ref{As_Factors} (\ref{As_Factors_ii}) and from the boundedness of $\Lambda$, with
\begin{eqnarray}
        B_{\Lambda}:= 2 \sup_{(t,x)\in(0,T)\times\mathbb{R}^n} |\Lambda^2(t,x)|
                                                                                                                                \nonumber
\end{eqnarray}
Thus,
\begin{eqnarray}\label{eq_JDRSAM_viscosity_RHSbound4-parabolic}
        \Big| F_b\left(t,y,\bar{v}(t,y),q,B\right)
        -  F_b\left(t,x,\bar{v}(t,y),p, A\right) \Big|
        \leq
        \omega\left( \frac{1}{\epsilon}|y-x|^2 + |y-x|\right)
\end{eqnarray}
where $\omega(z) = C_{\omega}z$ with $C_{\omega} := C_1+C_{\Lambda}$.\hfill\end{proof}

%
%

\begin{corollary}[Uniqueness]\label{coro_JDRSAM_valuefunc_uniquesol_PDE}
        The value function $\tilde{\Phi}$ is the unique viscosity solution of the parabolic PDE~\eqref{eq_JDRSAM_exptrans_HJBPDE_parabolic_outline} subject to terminal condition~\eqref{eq_JDRSAM_exptrans_HJBPDE_termcond}.
\end{corollary}

\begin{proof}
        Uniqueness of the viscosity solution follows from Theorem~\ref{theo_JDRSAM_comparison_parabolicPDE}. As we noted that  $\tilde{\Phi}$ is a viscosity solution of the parabolic PDE~\eqref{eq_JDRSAM_exptrans_HJBPDE_parabolic_outline}, it is therefore the unique solution to the boundary value problem.
\hfill\end{proof}

\subsection{Existence of a Classical Solution to the HJB PDE~\eqref{eq_JDRSAM_exptrans_HJBPDE_parabolic_outline}}\label{sec_classicalsol}

Now that we have been able to rewrite the HJB PIDE as a parabolic PDE, we have access to the literature addressing the existence of a strong solution to the HJB PDE, such as Fleming and Rishel~\cite{fl75} or Krylov~\cite{kr80} and~\cite{kr87}. The crucial point in this argument is that this new PDE has a unique viscosity solution which also solves the initial PIDE. From there, we only need to prove existence of a classical solution to the PDE in order to show that the value functions $\Phi$ is also of class $C^{1,2}((0,T)\times \mathbb{R}^n)$.
\\

The following ideas and notations are related to the treatment found in Ladyzenskaja, Solonnikov and Uralceva~\cite{la68} of linear parabolic partial differential equations. The relevant results are summarized in Appendix E of Fleming and Rishel \cite{fl75}. They concern PDEs of the form
\begin{eqnarray}\label{eq_JDRSAM_generic_linear_PDE}
        \frac{\partial \psi}{\partial t}
        + \frac{1}{2} \textrm{tr}
                    \left( a(t,x) D^2 \psi\right)
                +       b'(t,x) D\psi
        + \theta c(t,x) \psi
        + d(t,x)
    =   0
\end{eqnarray}
on a set $Q = (0,T) \times G$ and with boundary condition
\begin{eqnarray}
    \psi(t,x) &=& \Psi_T(x)    \quad x \in G
                            \nonumber\\
    \psi(t,x) &=& \Psi(t,x)    \quad (t,x) \in (0,T) \times \partial G
                            \nonumber
\end{eqnarray}
The set $G$ is open and is such that $\partial G$ is a compact manifold of class $C^2$. Denote by
\begin{itemize}
        \itemb $\partial^*Q$ the boundary of Q, i.e.
                \begin{equation}
                \partial^*Q := \left( \left\{ T \right\} \times G \right)
                                                                \cup \left( (0,T) \times \partial G \right)
                            \nonumber
                \end{equation}
        \itemb $L^{p}(K)$ the space of $p$-th power integrable functions on $K \subset Q$;
        \itemb $\lVert \cdot \rVert_{p,K}$ the norm in $L^{p}(K)$.
\end{itemize}
Also, denote by $\mathscr{H}^{p}(Q), 1 < p < \infty$ the space of all functions $\psi$ such that for $\psi(t,x)$ and all its generalized partial derivatives $\frac{\partial \psi}{\partial t}$, $\frac{\partial \psi}{\partial x_i}$, $\frac{\partial^2 \psi}{\partial x_ix_j}$, $i,j = 1,\ldots, n$  are in $L^{p}(K)$. We associate with this space the Sobolev-type norm:
\begin{eqnarray}
        \lVert \psi \rVert_{p,K}^{(2)}
        :=      \lVert \psi \rVert_{p,K}
        +       \Big\lVert \frac{\partial \psi}{\partial t} \Big\rVert_{p,K}
        +       \sum_{i=1}^n \Big\lVert \frac{\partial \psi}{\partial x_i} \Big\rVert_{p,K}
        +       \sum_{i,j=1}^n \Big\lVert \frac{\partial^2 \psi}{\partial x_i x_j} \Big\rVert_{p,K}
\end{eqnarray}
We will introduce additional notation and concepts as required in the proofs.
\\

Recall from the outline of Section 4 that~\eqref{eq_JDRSAM_exptrans_HJBPDE_parabolic_outline} is treated as a PDE with an autonomous term $d_a^{\tilde{\Phi}}$.
\newline

%
%
%
%
\begin{theorem}[Existence of a Classical Solution to the HJB PDE~\eqref{eq_JDRSAM_exptrans_HJBPDE_parabolic_outline}]\label{Theo_JDRSAM_existence}
        Under Assumption~\ref{as_factorjumps_xi_integrable} the RS HJB PDE~\eqref{eq_JDRSAM_exptrans_HJBPDE_parabolic_outline} with terminal condition $\tilde{u}(T,x) = 1$ has a solution $\tilde{u} \in C^{1,2}\left((0,T)\times\mathbb{R}^n\right)$ with $\tilde{u}$ continuous in $[0,T]\times\mathbb{R}^n$.
\end{theorem}
\newline

\begin{proof}
                The proof follows closely the argument of Fleming and Rishel~\cite{fl75} (Theorem VI.6.2 and Appendix E). 
\newline

\textbf{Step 1: Approximation in policy space - bounded space domain}\\
Consider the following auxiliary problem: fix $R>0$ and let $\mathscr{B}_R= \left\{x \in \mathbb{R}^n: |x|<R \right\}$. The HJB PDE for this auxiliary problem can be expressed as
\begin{eqnarray}\label{as_JDRSAM_auxiliaryPDE_tildePhi}
        &&      \frac{\partial u}{\partial t}
        + \frac{1}{2} \textrm{tr}
                    \left( \Lambda\Lambda'(t,x) D^2 u\right)
                +       H_{a}(t,x,u,Du)
                + d_{a}^{\tilde{\Phi}}(t,x)
    =   0
                                                \nonumber\\
        &&      \qquad \forall (t,x) \in Q_R := (0,T) \times \mathscr{B}_R
\end{eqnarray}
where
\begin{eqnarray}
   H_{a}(s,x,r,q) &=&
                \inf_{h \in \mathcal{J}} \left\{
        f_{a}(t,x,h)' q + \theta g(t,x,h) r
    \right\}
                                                                                                \nonumber
\end{eqnarray}
for $p \in \mathbb{R}^n$ and subject to boundary conditions
\begin{eqnarray}
        u(t, x) &=& \Psi(t,x)
        \qquad \forall (t,x) \in \partial^* Q_R := \left( (0,T) \times \partial \mathscr{B}_R \right) \cup \left( \left\{T\right\} \times \mathscr{B}_R \right)
                                                \nonumber
\end{eqnarray}
with
\begin{itemize}
\itemb $\Psi(T, x) = 1 \; \forall x \in  \mathscr{B}_R$;
\itemb $\Psi(t, x) =  \psi(t,x) \forall (t,x) \in (0,T) \times \partial \mathscr{B}_R$, where $\psi$ is some function of class $C^{1,2}(\overline{Q_R})$.
\end{itemize}
Define a sequence of functions $u^1$, $u^2$,... $u^k$,... on $\overline{Q_R}=[0,T]\times\overline{\mathscr{B}_R}$ and of bounded measurable feedback control laws $h^{0}$, $h^{1}$,... $h^{k}$,... where $h^{0}$ is an arbitrary control. $u^{k+1}$ solves the boundary value problem:
\begin{eqnarray}\label{eq_JDRSAM_logtrans_BVP_Phi_tilde_kplus1}
        \frac{\partial u^{k+1}}{\partial t}
        + \frac{1}{2} \textrm{tr}
                    \left( \Lambda\Lambda'(t,x) D^2 u^{k+1}\right)
                +       f_{a}(t,x,h^{k})' Du^{k+1} + \theta g(t,x,h^{k}) u^{k+1}
                + d_{a}^{\tilde{\Phi}}(t,x)
    =   0                                                                                       \nonumber\\
\end{eqnarray}
and subject to boundary conditions
\begin{eqnarray}
        u(t, x) &=& \Psi(t,x)
        \qquad \forall (t,x) \in \partial^* Q_R := \left( (0,T) \times \partial \mathscr{B}_R \right) \cup \left( \left\{T\right\} \times \mathscr{B}_R \right)
                                                \nonumber
\end{eqnarray}
Moreover, for almost all $(t,x) \in Q_R$, $k=1,2, \ldots$, we define $h^{k}$ by the prescription
\begin{eqnarray}\label{eq_JDRSAM_logtrans_PolicyImprovement_def_h_k}
    h^{k} = \textrm{Argmin}_{h \in \mathcal{J}}
        \left\{
                    f_{a}(t,x,h)' Du^{k}
        +   \theta g(t,x,h) u^{k}
        \right\}
\end{eqnarray}
so that
\begin{eqnarray}\label{eq_JDRSAM_logtrans_BVP_Phi_tilde_kplus1_identity}
        f_{a}(t,x,h^{k})' Du^{k}
        + \theta g(t,x,h^{k}) u^{k}
    &=& \inf_{h \in \mathcal{J}} \left\{
                    f_{a}(t,x,h)' Du^{k}
        +   \theta g(t,x,h) u^{k}
        \right\}
                                                \nonumber\\
    &=& H_{a}(t,x,u^{k},Du^{k})
\end{eqnarray}

Note that the boundary value problem~\eqref{eq_JDRSAM_logtrans_BVP_Phi_tilde_kplus1} is a special case of the generic problem introduced earlier in equation~\eqref{eq_JDRSAM_generic_linear_PDE} with
\begin{eqnarray}
    a(t,x)  &=& \Lambda\Lambda'(t,x)
                                                \nonumber\\
    b(t,x)  &=& f_{a}(t,x,h^{k})
                                                \nonumber\\
    c(t,x)  &=& g(t,x,h^{k})
                                                \nonumber\\
    d(t,x)  &=& d_{a}^{\tilde{\Phi}}(t,x)
                                                \nonumber
\end{eqnarray}
Moreover, since $\mathscr{B}_R$ is bounded and $\mathcal{J}$ is compact, all of these functions are also bounded. In particular $\lVert d_a^{\tilde{\Phi}} \rVert_{p,\mathscr{B}_R} < \infty$ for any $p>0$ because $\tilde{\Phi}$ is bounded and $d_a^{\tilde{\Phi}}$ is a continuous function of $\tilde{\Phi}$ by Lemma~\ref{lem_JDRSAM_d_bounded_cond2}. Thus, based on standard results on parabolic Partial Differential Equations (see for example Appendix E in Fleming and Rishel~\cite{fl75} and Chapter IV in Ladyzenskaja, Solonnikov and Uralceva~\cite{la68}), the boundary value problem~\eqref{eq_JDRSAM_logtrans_BVP_Phi_tilde_kplus1} admits a unique solution in $\mathscr{H}^{p}(Q_R)$ for any $p > 0$ depending only on the properties of $d$ and on the boundary conditions. Indeed, by estimate (E.8) in Appendix E of Fleming and Rishel \cite{fl75},
\begin{eqnarray}\label{eq_Sobolev_2_norm}
        \lVert u^k \rVert_{p,Q_R}^{(2)}
        \leq M_R^1 \left(
                \lVert d_a^{\tilde{\Phi}}(t,x) \rVert_{p,Q_R}
        +       \lVert \Psi \rVert_{p,\partial*Q_R}^{(2)}
        \right)
\end{eqnarray}
for some constant $M_R^1$ depending on $R$. The terms $\lVert d_a^{\tilde{\Phi}}(t,x) \rVert_{p,Q_R}$ and $\lVert \Psi \rVert_{p,\partial*Q_R}^{(2)}$ are finite which implies that $\lVert u^k \rVert_{p,Q_R}^{(2)}$ is finite. Thus the sequence $(u^k)_{k \in \mathbb{N}}$ is bounded.
\newline

Moreover, if we pick $p > \frac{n+2}{2}$, the boundedness of $ \lVert u^k \rVert_{p,Q_R}^{(2)}$ resulting from estimate~\eqref{eq_Sobolev_2_norm} implies that the H\"older norm 
\begin{eqnarray}
        \lvert u^k \rvert_{Q_R}^{\alpha}
    &=& \sup_{(t,x)\in Q_R} \lvert u^k (t,x) \rvert
    +   \sup_{\begin{array}{c} (x,y)\in \overline{G} \\ 0 \leq t \leq T \end{array}}
            \frac{\lvert u^k(t,x)- u^k(t,y) \rvert}{\lvert x - y \rvert^{\alpha}}
                                            \nonumber\\
    &+& \sup_{\begin{array}{c} x \in \overline{G} \\ 0 \leq s,t \leq T \end{array}}
            \frac{\lvert u^k(s,x)- u^k(t,x) \rvert}{\lvert s - t \rvert^{\alpha/2}}
                                            \nonumber
\end{eqnarray}
is bounded for some $0<\alpha<1$, proving that $u^k$ is continuous on $\overline{Q_{R}}$ (see Ladyzenskaja, Solonnikov and Uralceva~\cite{la68}), p.89). Define the H\"older norm $\lvert u^k \rvert_{Q_R}^{1,\alpha}$ by 
 \begin{eqnarray}
        \lvert u^k \rvert_{Q_R}^{1,\alpha}
    =   \lvert u^k \rvert_{Q_R}^{\alpha}
    +   \sum_{i=1}^{n} \lvert u_{x_i}^k \rvert_{Q_R}^{\alpha}
                                            \nonumber
\end{eqnarray}
For $p > n+2$, we have the following estimate (see equation (E.9) in Appendix E of Fleming and Rishel \cite{fl75} ) 
\begin{eqnarray}\label{eq_JDRSAM_logtrans_BVP_bound0}
            \lvert u^k \rvert_{Q_R}^{1,\alpha}
    \leq    M_R^2 \lVert u^k \rVert_{p,Q_R}^{(2)}
\end{eqnarray}
for some constant $M_R^2$ (depending on $R$) and with
\begin{equation}
    \alpha = 1 -\frac{n+2}{p}
                                            \nonumber
\end{equation}
The H\"older norm $\lvert u^k \rvert_{Q_R}^{1,\alpha}$ is finite and as a result $Du^k$ is continuous on $\overline{Q_{R}}$. In particular, if we denote by $\mathscr{C}^{1,\alpha}(Q), 0<\alpha<1$ the H\"older space of all functions $\psi$ such that $\lvert \psi \rvert_{Q}^{1,\alpha}$ is finite, then $u^k \in \mathscr{C}^{1,\alpha}(Q_R)$.
\newline

\textbf{Step 2: Convergence Inside the Cylinder $(0,T) \times \mathscr{B}_R$}\\
Take $k \geq 1$. By~\eqref{eq_JDRSAM_logtrans_BVP_Phi_tilde_kplus1_identity},
\begin{eqnarray}
                &&
        \frac{\partial u^{k}}{\partial t}
        + \frac{1}{2} \textrm{tr}
                    \left( \Lambda\Lambda'(t,x) D^2 u^{k}\right)
                +       f_{a}(t,x,h^{k})' Du^{k} + \theta g(t,x,h^{k}) u^{k}
                + d_{a}^{\tilde{\Phi}}(t,x)
                                                                                                        \nonumber\\
                &\leq&
        \frac{\partial u^{k}}{\partial t}
        + \frac{1}{2} \textrm{tr}
                    \left( \Lambda\Lambda'(t,x) D^2 u^{k}\right)
                +       f_{a}(t,x,h^{k-1})' Du^{k} + \theta g(t,x,h^{k-1}) u^{k}
                + d_{a}^{\tilde{\Phi}}(t,x)
                                                                                                        \nonumber\\
    &=& 0
                                                                                                        \nonumber
\end{eqnarray}
Subtracting~\eqref{eq_JDRSAM_logtrans_BVP_Phi_tilde_kplus1},
\begin{eqnarray}\label{eq_JDRSAM_logtrans_BVP_Phi_tilde_kplus2}
                &&
        \left(
                                        \frac{\partial u^{k+1}}{\partial t}
                                -       \frac{\partial u^{k}}{\partial t}
                        \right)
        + \frac{1}{2} \textrm{tr}
                    \left( \Lambda\Lambda'(t,x) D^2 \left[
                                                                        u^{k+1}
                                                        -       u^{k}
                                                        \right]\right)
                                                                                                        \nonumber\\
                &&
                +       f_{a}(t,x,h^{k})'
                                \left(Du^{k+1}-Du^{k+1}\right)
                + \theta g(t,x,h^{k})
                                \left(u^{k+1} - u^{k}\right)
                \leq    0
\end{eqnarray}

Define the sequence of functions $(W^k)_{k \in \mathbb{N}}$ as
\begin{eqnarray}
        W^k := u^{k+1} - u^{k}
                                                                        \nonumber
\end{eqnarray}
then $W^k$ satisfies the inequality
\begin{eqnarray}\label{as_JDRSAM_intermediatePDE1}
    &&
        \frac{\partial W^{k}}{\partial t}
      +  \frac{1}{2} \textrm{tr} \left(\Lambda\Lambda'(t,x)D^2 W^{k}\right)
                +  f_{a}(t,x,h^{k})'DW^{k}
      +  \theta g(t,x,h^{k}) W^{k}
                \leq 0
\end{eqnarray}
in $(0,T)\times\mathscr{B}_R$, and with boundary condition $W^{k}(T,x) = 0$ on $\partial^* Q_R = \left( (0,T) \times \partial\mathscr{B}_R \right) \cup \left( \left\{T\right\} \times \mathscr{B}_R \right)$. Moreover, $W^k$ is continuous on $\overline{Q_R}$ by continuity of $u^k$ $\forall k \in \mathbb{N}$. By the Maximum Principle in Lemma~\ref{theo_classical_weak_maximum_principle_extensionboundedcost} below, $W^k(t,x) \leq 0$ for $k \geq 1$ and hence by definition of $W^k$,
\begin{eqnarray}
        u^{k} \geq u^{k+1}
        , \qquad \forall k \in \mathbb{N}
                                                                        \nonumber
\end{eqnarray}
which implies that the sequence $\left\{ u^{k} \right\}_{k \in \mathbb{N}}$ is non decreasing. As a a result, it is bounded from above by $u^1$. Since the sequence $(u^k)_{k \in \mathbb{N}}$ is non-increasing and is also bounded, it converges. Moreover the convergence is uniform because $u^k$ is uniformly continuous.
Denote by $u$ its limit as $k \to \infty$. 
\\

For $p > n+2$, $Du^k$ is continuous and as $k \to \infty$ we conclude that
\begin{itemize}
\itemb $Du^k$ converges to $Du$ uniformly on $\overline{Q_R}$ ;
\itemb $D^2u^k$, $\frac{\partial u^k}{\partial t}$ converge weakly in $L^{p}(Q_R)$ respectively to $D^2u$ and $\frac{\partial u}{\partial t}$.
\newline
\end{itemize}

\indent\indent\textbf{Step 3: Proving that $u \in C^{1,2}(Q_R)$}\\
Using relationship~\eqref{eq_JDRSAM_logtrans_BVP_Phi_tilde_kplus1_identity} and then equation~\eqref{eq_JDRSAM_logtrans_BVP_Phi_tilde_kplus1}, we get:
\begin{eqnarray}\label{eq_JDRSAM_logtrans_BVP_ineq1}
    &&
        \frac{\partial u^{k}}{\partial t}
        + \frac{1}{2} \textrm{tr}
                    \left( \Lambda\Lambda'(t,x) D^2 u^{k}\right)
                +       f_{a}(t,x,h)' Du^{k} + \theta g(t,x,h) u^{k}
                + d_{a}^{\tilde{\Phi}}(t,x)
                                                \nonumber\\
    &\geq&
        \frac{\partial u^{k}}{\partial t}
        + \frac{1}{2} \textrm{tr}
                    \left( \Lambda\Lambda'(t,x) D^2 u^{k}\right)
                +       f_{a}(t,x,h^{k})' Du^{k} + \theta g(t,x,h^{k}) u^{k}
                + d_{a}^{\tilde{\Phi}}(t,x)
                                                \nonumber\\
    &=&
        \left( \frac{\partial u^{k}}{\partial t} - \frac{\partial u^{k+1}}{\partial t} \right)
        +   \left(
            \frac{1}{2} \textrm{tr} \left[
                \left( \Lambda\Lambda'(t,x)) D^2 u^{k}\right)
                -   \left( \Lambda\Lambda'(t,x) D^2 u^{k+1}\right)
                \right]\right.
                                                \nonumber\\
    &&  \left.
                +   f_{a}(t,x,h^{k})' \left( Du^{k} - Du^{k+1} \right)
        +   \theta g(t,x,h^{k}) \left( u^{k} - u^{k+1} \right)
        \right]
\end{eqnarray}
for any admissible control $h$. Since the left-hand side of~\eqref{eq_JDRSAM_logtrans_BVP_ineq1} tends weakly in $L^{p}(Q_R)$ to
\begin{eqnarray}\label{eq_JDRSAM_proofstep2_ineq1}
        \frac{\partial u}{\partial t}
        + \frac{1}{2} \textrm{tr}
                    \left( \Lambda\Lambda'(t,x) D^2 u \right)
                +       f_{a}(t,x,h)' Du + \theta g(t,x,h) u
                + d_{a}^{\tilde{\Phi}}(t,x)
\end{eqnarray}
as $k \to \infty$ and the right-hand side tends tends weakly to 0, then we obtain the following inequality
\begin{eqnarray}
        \frac{\partial u}{\partial t}
        + \frac{1}{2} \textrm{tr}
                    \left( \Lambda\Lambda'(t,x) D^2 u \right)
                +       f_{a}(t,x,h)' Du + \theta g(t,x,h) u
                + d_{a}^{\tilde{\Phi}}(t,x)
    \geq    0
                                                \nonumber
\end{eqnarray}
almost everywhere in $Q_R$.
\\

Using a measurable selection theorem and following an argument similar to that of Lemma VI.6.1 of Fleming and Rishel~\cite{fl75}, we see that there exists a Borel measurable function $h^{*}$ from $(0,T)\times\mathscr{B}_{R}$ into $\mathcal{J}$ such that.
\begin{eqnarray}
        f_{a}(t,x,h^*)' Du
        + \theta g(t,x,h^*) u
    &=& \inf_{h \in \mathcal{J}} \left\{
                    f_{a}(t,x,h)' Du
        +   \theta g(t,x,h) u\right\}
                                                                                                                                \nonumber
\end{eqnarray}
holds for almost all $(t,x) \in (0,T)\times\mathscr{B}_{R}$. Then
\begin{eqnarray}\label{eq_JDRSAM_logtrans_BVP_ineq2}
    &&
        \frac{\partial u^{k}}{\partial t}
        + \frac{1}{2} \textrm{tr}
                    \left( \Lambda\Lambda'(t,x) D^2 u^{k}\right)
                +       f_{a}(t,x,h^{*})' Du^{k} + \theta g(t,x,h^{*}) u^{k}
                + d_{a}^{\tilde{\Phi}}(t,x)
                                                \nonumber\\
    &\leq&
        \frac{\partial u^{k}}{\partial t}
        + \frac{1}{2} \textrm{tr}
                    \left( \Lambda\Lambda'(t,x) D^2 u^{k}\right)
                +       f_{a}(t,x,h^{k})' Du^{k} + \theta g(t,x,h^{k}) u^{k}
                + d_{a}^{\tilde{\Phi}}(t,x)
                                                \nonumber\\
    &=&
        \left( \frac{\partial u^{k}}{\partial t} - \frac{\partial u^{k+1}}{\partial t} \right)
        +   \left(
            \frac{1}{2} \textrm{tr} \left[
                \left( \Lambda\Lambda'(t,x) D^2 u^{k}\right)
                -   \left( \Lambda\Lambda'(t,x) D^2 u^{k+1}\right)
                \right]\right.
                                                \nonumber\\
    &&  \left.
                +   f_{a}(t,x,h^{k})' \left( Du^{k} - Du^{k+1} \right)
        +   \theta g(t,x,h^{k}) \left( u^{k} - Du^{k+1} \right)
        \right]
\end{eqnarray}

Since the left-hand side of~\eqref{eq_JDRSAM_logtrans_BVP_ineq2} tends weakly in $L^{p}(Q_R)$ to
\begin{eqnarray}
        \frac{\partial u}{\partial t}
        + \frac{1}{2} \textrm{tr}
                    \left( \Lambda\Lambda'(t,x) D^2 u \right)
                +       f_{a}(t,x,h^{*})' Du + \theta g(t,x,h^{*}) u
                + d_{a}^{\tilde{\Phi}}(t,x)
                                                \nonumber
\end{eqnarray}
as $k \to \infty$ and the right-hand side tends weakly to 0, then we obtain the inequality
\begin{eqnarray}\label{eq_JDRSAM_proofstep2_ineq2}
        \frac{\partial u}{\partial t}
        + \frac{1}{2} \textrm{tr}
                    \left( \Lambda\Lambda'(t,x) D^2 u \right)
                +       f_{a}(t,x,h^{*})' Du + \theta g(t,x,h^{*}) u
                + d_{a}^{\tilde{\Phi}}(t,x)
    \leq    0
\end{eqnarray}
almost everywhere in $Q_R$.
\\

Combining~\eqref{eq_JDRSAM_proofstep2_ineq1} and~\eqref{eq_JDRSAM_proofstep2_ineq2}, we have therefore shown that
\begin{eqnarray}\label{eq_PDE_step3_1}
        \frac{\partial u}{\partial t}
        + \frac{1}{2} \textrm{tr}
                    \left( \Lambda\Lambda'(t,x) D^2 u \right)
                +       f_{a}(t,x,h^{*})' Du + \theta g(t,x,h^{*}) u
                + d_{a}^{\tilde{\Phi}}(t,x)
    =   0
\end{eqnarray}
almost everywhere in $Q_R$. Hence, $u$ is a solution of equation~\eqref{as_JDRSAM_auxiliaryPDE_tildePhi}. Moreover, $u \in \mathscr{H}^{p}(Q_R)$ and as a result $u \in \mathscr{C}^{1,\alpha}(Q_R)$ for $p>n+2$ by estimate~\eqref{eq_JDRSAM_logtrans_BVP_bound0}. Also, since $H_a$ is locally Lipschitz, $\lvert u \rvert_{Q_R}^{\alpha} < \infty$ for $\alpha > 0$ and $\lvert Du \rvert_{Q_R}^{\alpha} < \infty$ for $0 < \alpha \leq 1$, then $\lvert H_a(t,x,u,Du) \rvert_{Q_R}^{\alpha} < \infty$. 
\newline

The key idea now is that we can rewrite PDE~\eqref{eq_PDE_step3_1} as
\begin{eqnarray}\label{eq_PDE_step3_2}
        \frac{\partial u}{\partial t}
        + \frac{1}{2} \textrm{tr}
                    \left( \Lambda\Lambda'(t,x) D^2 u \right)
                +       H_{a}(t,x,u,Du)
                + d_{a}^{\tilde{\Phi}}(t,x)
    =   0
\end{eqnarray}
which only depends on $t$ and $x$, but not explicitly on $h$.
Let $H_a^u(t,x) := H_{a}(t,x,u,Du)$ and consider the auxiliary PDE
\begin{eqnarray}\label{eq_PDE_step3_3}
        \frac{\partial v}{\partial t}
        + \frac{1}{2} \textrm{tr}
                    \left( \Lambda\Lambda'(t,x) D^2 v \right)
                +       F(t,x)
    =   0
\end{eqnarray}
where $F(t,x) := H_{a}^u(t,x) + d_{a}^{\tilde{\Phi}}(t,x)$, and subject to boundary conditions
\begin{eqnarray}
        v(t, x) &=& \Psi(t,x)
        \qquad \forall (t,x) \in \partial^* Q_R := \left( (0,T) \times \partial \mathscr{B}_R \right) \cup \left( \left\{T\right\} \times \mathscr{B}_R \right)
                                                \nonumber
\end{eqnarray}
It is clear that $u$ solves both~\eqref{eq_PDE_step3_2} and~\eqref{eq_PDE_step3_3}. Using similar arguments to those developed so far, we can show that auxiliary PDE~\eqref{eq_PDE_step3_3} admits a unique solution $v$, which coincides with $u$.
\newline

To go further, we need to know that $F$ is H\"older continuous. However, $d_a^{\tilde{\Phi}}$ is continuous by Lemma~\ref{lem_JDRSAM_d_bounded_cond2} and $H_a$ is locally Lipschitz by Remark~\ref{rk_JDRSAM_Ha_Lipschitz}, so we can use the argument proposed by Davis, Guo and Wu~\cite{guwu09} in the proof of their Theorem 5.5 to show that $F$ is H\"older continuous in $x$ with exponent $0<\alpha<1$. We omit the details of the proof here as it is lengthy but solely requires standard calculus techniques.
\newline

We can now show that $u \in C^{1,2}(Q_R)$. Define
\begin{eqnarray}
        \lvert v^k \rvert_{Q_R}^{2,\alpha}
   :=   \lvert v^k \rvert_{Q_R}^{1,\alpha}
        +       \Big\lvert \frac{\partial v^k}{\partial t} \Big\rvert_{Q_R}^{\alpha}
   +    \sum_{i,j=1}^{n} \lvert v_{x_i x_j}^k \rvert_{Q_R}^{\alpha}
                                            \nonumber
\end{eqnarray}
and denote by $\mathscr{C}^{2,\alpha}(Q), 0<\alpha<1$ the H\"older space of all functions $\psi$ such that $\lvert \psi \rvert_{Q}^{2,\alpha}$ is bounded.
\newline

Now, consider two open subsets $Q'$ and $Q''$ of $Q$ such that $\bar{Q'} \subset \bar{Q''}$. By estimate (E10) in Appendix E of Fleming and Rishel \cite{fl75}, we have
\begin{eqnarray}\label{eq_JDRSAM_logtrans_BVP_bound1}
            \lvert v \rvert_{Q'}^{2,\alpha}
    \leq    M_2 \left(
                                     |F|_{Q''}^{\alpha}
                             +  \lVert v \rVert_{Q''} \right)
\end{eqnarray}
for some constant $M_2$ depending solely on $Q'$ and $Q''$. Set $Q'' = Q_R$ and take $Q'$ to be any subset of $Q$ such that $\bar{Q}' \subset Q$. Thus
\begin{eqnarray}\label{eq_JDRSAM_logtrans_BVP_bound2}
            \lvert u \rvert_{Q'}^{2,\alpha} 
                =       \lvert v \rvert_{Q'}^{2,\alpha} 
                \leq    M_2 \left( |F|_{Q_R}^{\alpha} + \lVert v \rVert_{Q_R} \right)
                        <       \infty
\end{eqnarray}
When interpreted in light of estimate~\eqref{eq_JDRSAM_logtrans_BVP_bound0}, we see that the derivatives $\frac{\partial u}{\partial t}$, $\frac{\partial u}{\partial x_i}$ and $\frac{\partial^2 u}{\partial x_i x_j}$ satisfy a uniform H\"older condition on any compact subset $Q'$ of $Q_R$. By Theorem 10.1 in Chapter IV of Ladyzenskaja, Solonnikov and Uralceva~\cite{la68}, we can therefore conclude that $u \in C^{1,2}(Q_R)$.
\\

\textbf{Step 4: Convergence from the Cylinder $(0,T) \times \mathscr{B}_R$ to the State Space $(0,T) \times \mathbb{R}^n$}\\

For $l=1,2,\ldots$, we define a function $\alpha_l$ satisfying the following
\begin{enumerate}[(i).]
\item $\alpha_l \in C^{\infty}$;
\item $\alpha_l \geq 0$;
\item
        \begin{eqnarray}
                        \alpha_l(x) :=
                        \left\{ \begin{array}{ll}
                                1       &       \textrm{for } x \in \mathscr{B}_{l}             \\
                                0       &       \textrm{for } x \in \mathbb{R}^n\backslash\mathscr{B}_{l+1}     \\
                        \end{array}\right.
                                                                                                                                \nonumber
        \end{eqnarray}
\item   $\big| \frac{\partial \alpha_{l}}{\partial x} \big| \leq 2$.
\end{enumerate}
Let $\tilde{\Phi}_l$ be a solution of the PDE:
\begin{eqnarray}\label{as_JDRSAM_auxiliaryPDE_tildePhi_step3}
        \frac{\partial u_l}{\partial t}
        + \frac{1}{2} \textrm{tr}
                    \left( \Lambda\Lambda'(t,x) D^2 u_l\right)
                +       \alpha_l(x)H_{a}\left(t,x,u_l,Du_l\right)
                + d_{a}^{\tilde{\Phi}}(t,x)
    =   0
\end{eqnarray}
and subject to terminal condition $u_l(T, x) = \alpha_l(x)$.
\\

By the local estimate~\eqref{eq_JDRSAM_logtrans_BVP_bound0} from step 1 above, $\lVert u_l \rVert_{p,Q_R}^{(2)}$ is bounded for $p >1$ for any bounded $Q \subset (0,T)\times\mathbb{R}^n$ and $W_l$ satisfies a H\"older condition on each bounded $Q$. In particular, take $Q := (0,T)\times\mathscr{B}_{l_0}$ for some $l_0 >0$. For $l>l_0$, $u_l$ solves the ``parabolic'' PDE~\eqref{eq_JDRSAM_exptrans_HJBPDE_parabolic_outline} in $Q$.
\\

By Remark~\ref{rk_JDRSAM_Ha_Lipschitz}, $H_a$ is locally Lipschitz continuous. Taking into account the estimate~\eqref{eq_JDRSAM_logtrans_BVP_bound1}, we find, that $\frac{\partial u^{(i)}}{\partial t}$ and $\frac{\partial^2 u^{(i)}}{\partial x_i x_j}$ also satisfy a uniform H\"older condition on any compact subset of $Q$.
\\

By Ascoli's theorem, we can find a subsequence $\left(u^l\right)_{l \in \mathbb{N}}$ of $\left(u^{(i)}\right)_{i \in \mathbb{N}}$ such that \begin{enumerate}[(i).]
    \itemb $\left(u^l \right)_{l \in \mathbb{N}}$ tends to a limit $\tilde{u}$ uniformly on each compact subset of $\overline{Q}_0$;
    \itemb $\left(\frac{\partial u}{\partial t}^l \right)_{l \in \mathbb{N}}$ tends to a limit $\frac{\partial \tilde{u}}{\partial t}$ uniformly on each compact subset of $Q_0$;
    \itemb $\left(Du^l \right)_{l \in \mathbb{N}}$ tends to a limit $D\tilde{u}$ uniformly on each compact subset of $Q_0$;
    \itemb $\left(D^2u^l \right)_{l \in \mathbb{N}}$ tends to a limit $D^2\tilde{u}$ uniformly on each compact subset of $Q_0$.
\end{enumerate}

To conclude, the function $\tilde{u}$ is a classical solution of the equation~\eqref{eq_JDRSAM_exptrans_HJBPDE_parabolic_outline} with terminal condition $\tilde{\Phi}(T,x) = 1$.

\end{proof}

A parabolic version of the weak maximum principle is required in the previous proof to show that the sequence of function $(u^k)_{k=1,\ldots}$ converges.\footnote{Note that we do not need the maximum principle to show that $\tilde{\Phi} \geq 0$ as this is a direct consequence of the structure of risk-sensitive control.} The parabolic weak maximum principle used in the proof takes the form:

\begin{lemma}[Classical Weak Maximum Principle (Theorem 5.1 in Pham~\cite{ph98})]\label{theo_classical_weak_maximum_principle}
Consider a parabolic PDE of the form~\eqref{eq_JDRSAM_generic_linear_PDE} with parabolic operator
\begin{eqnarray}\label{eq_JDRSAM_diffstate_generic_parabolic_operator}
                L\psi &=&
        \frac{\partial \psi}{\partial t}
        + \frac{1}{2} \textrm{tr}
                    \left( a(t,x) D^2 \psi\right)
                +       b(t,x)' D\psi
      + \theta c(t,x) \psi
                + \frac{\partial \psi}{\partial t}
\end{eqnarray}
and assume
\begin{itemize}
\itemb $Q$ is bounded;
\itemb the operator $L\psi$ is parabolic for every $(x,t) \in Q$;
\itemb the coefficients $a(t,x)$, $b(t,x)$, $c(t,x)$ and $d(t,x)$ are continuous;
\itemb $c(t,x) \leq 0$;
\end{itemize}
If
\begin{itemize}
\itemb $\psi$ is continuous on $\overline{Q}$ in the viscosity sense,
\itemb $L\psi \geq 0$ in Q,
\itemb $\psi \leq 0$ on $\partial^* Q$
\end{itemize}
then
$\psi \leq 0$ on $\bar{Q}$.
\end{lemma}

\begin{proof}
        Refer to Pham~\cite{ph98}.
\end{proof}

The condition $c(x) \leq 0$ is proving too restrictive for our purpose. The following classical extension to the weak maximum principle partially addresses this difficulty:

\begin{lemma}[Extension of Theorem~\ref{theo_classical_weak_maximum_principle} for a bounded cost function $c(x,t)$ (Theorem 5.1 in Pham~\cite{ph98})]\label{theo_classical_weak_maximum_principle_extensionboundedcost}

Consider a parabolic PDE of the form~\ref{eq_JDRSAM_generic_linear_PDE} and assume
\begin{itemize}
\itemb $Q$ is bounded;
\itemb the operator $L\psi$ is parabolic for every $(x,t) \in Q$;
\itemb the coefficients $a(t,x)$, $b(t,x)$, $c(t,x)$ and $d(t,x)$ are continuous;
\itemb $c(t,x) \leq c_0$, where $c_0 \geq 0$;
\end{itemize}
If
\begin{itemize}
\itemb $\psi$ is continuous on $\overline{Q}$ in the viscosity sense,
\itemb $L\psi \geq 0$ in Q,
\itemb $\psi \leq 0$ on $\partial^* Q$
\end{itemize}
then
$\psi \leq 0$ on $Q$.

\end{lemma}

\begin{proof}
Introduce the function $u(t,x) := e^{-c_0t}\psi(t,x)$. The second set of assumptions for the function $\psi$ translates for $u$ into:
\begin{itemize}
\itemb $u \in C^{1,2}(Q)$,
\itemb $(L-c_0)u \geq 0$ in Q,
\itemb $u$ is continuous on $\overline{Q}$, and
\itemb $u \leq 0$ on $\partial^* Q$
\end{itemize}
Moreover the coefficient of $u(t,x)$ in the parabolic operator for $u$ is now $c(t,x)-c_0 \leq 0$. Hence, we can apply Theorem~\ref{theo_classical_weak_maximum_principle} to the function $u$ in order to deduce that  $u \leq 0$ on $Q$ and therefore $\psi \leq 0$ on $Q$.

\end{proof}

\begin{remark}
The control-based argument used in the proof of Theorem 7.2 in~\cite{dall_JDRSAM_Diff} cannot be used here. The reason for this is that reinterpreting the PIDE as a PDE removes the natural connection between the PDE and the dynamics of the factor process. It therefore becomes more effective to consider the PDE in abstraction from the control problem and prove existence directly through standard PDE arguments.
\end{remark}

\subsection{Existence of a Classical Solution to the HJB PIDE~\eqref{eq_JDRSAM_exptrans_HJBPDE}} 
So far, we have shown that $\tilde{\Phi}$ is a viscosity solution to HJB PIDE~\eqref{eq_JDRSAM_exptrans_HJBPDE} in Section 4.1 and that $\tilde{u}$ is a classical solution of the equation~\eqref{eq_JDRSAM_exptrans_HJBPDE} with terminal condition $\tilde{\Phi}(T,x) = 1$ in Section 4.5. In this subsection, we prove that HJB PIDE~\eqref{eq_JDRSAM_exptrans_HJBPDE} admits a classical solution by showing that $\tilde{u} = \tilde{\Phi}$.
\newline

\begin{theorem}[Existence of a Classical Solution to the HJB PIDE~\eqref{eq_JDRSAM_exptrans_HJBPDE}]\label{Theo_JDRSAM_existence_PIDE}
        Under Assumption~\ref{as_factorjumps_xi_integrable} the RS HJB PIDE~\eqref{eq_JDRSAM_exptrans_HJBPDE} with terminal condition $\tilde{\Phi}(T,x) = 1$ has a unique solution $\tilde{\Phi} \in C^{1,2}\left((0,T)\times\mathbb{R}^n\right)$ with $\tilde{\Phi}$ continuous in $[0,T]\times\mathbb{R}^n$.
\end{theorem}

\begin{proof}
We showed in Theorem~\ref{theo_JDRSAM_viscositysol}, Proposition~\ref{prop_PIDE_to_PDE} and Theorem~\ref{theo_JDRSAM_comparison_parabolicPDE} that $\tilde{\Phi}$ is both a Lipschitz continuous viscosity solution of HJB PIDE~\eqref{eq_JDRSAM_exptrans_HJBPDE} and the unique viscosity solution of HJB PDE~\eqref{eq_JDRSAM_exptrans_HJBPDE_parabolic_outline}. Now, observe that a classical solution is also a viscosity solution. Broadly speaking the argument is that if the solution of the PDE is smooth, then we can use it as a test function in the definition of viscosity solutions. If we do this, we will recover the classical maximum principle and therefore prove that the solution of the PDE is a classical solution. Hence, the classical solution $\tilde{u}$ to PDE~\eqref{eq_JDRSAM_exptrans_HJBPDE_parabolic_outline}, whose existence was proved in Theorem~\ref{Theo_JDRSAM_existence}, is also a viscosity solution of PDE~\eqref{eq_JDRSAM_exptrans_HJBPDE_parabolic_outline}.
As a result, $\tilde{\Phi} = \tilde{u}$ and we conclude that $\tilde{\Phi}$ is $C^{1,2}$ and satisfies PIDE~\eqref{eq_JDRSAM_exptrans_HJBPDE} in the classical sense. Uniqueness follows from Corollary \ref{coro_JDRSAM_valuefunc_uniquesol_PDE}. \hfill\end{proof}
\begin{corollary}[Existence of a Classical Solution for the Risk-Sensitive Control Problem]\label{Coro_JDRSAM_existence_PIDE}
        The RS HJB PIDE~\eqref{eq_JDRSAM_HJBPDE} with terminal condition $\Phi(T,x) =0$ has a unique solution $\Phi \in C^{1,2}\left([0,T]\times\mathbb{R}^n\right)$ with $\Phi$ continuous in $[0,T]\times\mathbb{R}^n$.
\end{corollary}
\proof This follows from the basic relationship between $\Phi$ and $\tilde{\Phi}$, i.e. $\tilde{\Phi}=\exp\{-\theta\Phi\}$.\hfill\endproof

%
%
\section{Identifying the Optimal Strategy}\label{optimalcontrol} All that remains now is to show that the controls derived from the Hamiltonian-minimizing function $\check{h}$ of Corollary \ref{coro_JDRSAM_optimcontrol_tilde} and from the maximiser $\hat{h}$ of Proposition \ref{prop_JDRSAM_optimcontrol} correspond to the optimal policy.\\

\begin{lemma}\label{lem_JDRSAM_hstar} 
The Hamiltonian-minimizing function $\check{h}$ of Corollary \ref{coro_JDRSAM_optimcontrol_tilde} and the maximiser $\hat{h}$ of Proposition \ref{prop_JDRSAM_optimcontrol} correspond to the same control $h^*(t,X_t)$, that is:
$$\hat{h}(t,X_t,D\Phi(t,X_t)) = \check{h}(t,X_t,\tilde{\Phi}(t,X_t),D\tilde{\Phi}(t,X_t)) = : h^*(t,X_t)$$.
\end{lemma}

\begin{proof}
        The Hamiltonian-minimizing function $\check{h}$ of Corollary \ref{coro_JDRSAM_optimcontrol_tilde} and the maximiser $\hat{h}$ of Proposition \ref{prop_JDRSAM_optimcontrol} are both unique. Moreover $\Phi$ and $\tilde{\Phi}$ are related through a monotone transformation.This proves that $\hat{h}(t,X_t,D\Phi(t,X_t)) = \check{h}(t,X_t,\tilde{\Phi}(t,X_t),D\tilde{\Phi}(t,X_t))$.\hfill
\end{proof}\\

\begin{remark}
	The control $h^*$ introduced in step 3 of the proof of Theorem~\ref{Theo_JDRSAM_existence} is the control $h^*(t,X_t)$.
\end{remark}
\hfill

\begin{proposition}\label{prop_JDRSAM_hstar_admissible} The control $h^*(t,X_t)$ is admissible: $h^*(t,X_t) \in \mathcal{A}$.
\end{proposition}
\begin{proof}
        The proof follows closely Proposition 4.3 in Davis and Lleo~\cite{dall_JDRSAM_Diff}. The class of admissible controls is presented in Definition~\ref{def_JDRSAM_admissible_A}.\hfill
\end{proof}\\

\begin{remark}
        The argument used in~\cite{dall_JDRSAM_Diff} is based on a result by M\'emin~\cite{me79}. One could also derive a similar argument using the elegant results of~\cite{chfiyo05}. 
\end{remark}\\

\begin{theorem}\label{Theo_optimalcontrol}
The control $h^*$ is optimal. In particular $\tilde{\Phi}(t,x) = \tilde{I}(v,x,h^*;t;T;\theta)$.
\end{theorem}

\begin{proof}
Consider the Borel-measurable minimizing control $h^*(t,X_t)$ with associated measure $\mathbb{P}^*$ and let $X(s), s\geq t$ be the state process with initial data $X(t) = x$.

Define the process $Z(s) := \theta \int_{t}^{s}g(u,X_u,h^*_u) du$ and use the general Ito formula to calculate $Z(s)\tilde{\Phi}(s, X(s))$. We find that 
\begin{eqnarray}
    \tilde{\Phi}(s,X_s)e^{Z_s} &=&
        \tilde{\Phi}(t,x)
        + \int_{t}^{s} D\tilde{\Phi}'\Lambda(u,X(u)) dW_{u}^{\theta}
                                                                                                                                        \nonumber\\
        &&
    +         \int_t^s  \int_{\mathbf{Z}} \left\{
                \tilde{\Phi}\left(u,X(u^-)+\xi\left(u,X(u^-),z\right)\right)
            -   \tilde{\Phi}(u,X(u^-))
        \right\} \tilde{N} (du,dz).\label{**}
\end{eqnarray}
(The drift term is equal to zero since $h^*$ achieves the minimum in the HJB PIDE~\eqref{eq_JDRSAM_exptrans_HJBPDE}.) We claim that both stochastic integrals in \eqref{**} are martingales. Indeed, under Assumption~\ref{As_Factors}, Theorem 1.19 of \cite{oksu05} implies that
\begin{equation} \mathbf{E}^*\left[\int_t^T|X(s)|^2ds\right]<\infty.\label{***}\end{equation}
This is enough to show that the Brownian integral is a martingale, since $D\tilde{\Phi}$ is bounded and $\Lambda$ is Lipschitz in $x$. For the Poisson random measure integral we have, since $\tilde{\Phi}$ is Lipschitz (with constant $K$),
\[ \alpha(u,X(u^-),z)\equiv|\tilde{\Phi}(u,X(u^-)+\xi(u,X(u^-),z)-\tilde{\Phi}(u,X(u^-))|\leq K|\xi(u,X(u^-),z)|,\]
so, by Assumption \eqref{as_factorjumps_xi_sq_int},
\[ \int_\mathbf{Z}\alpha^2(u,X(u^-),z)\nu(dz)<3c^2(1+|X(u^-)|^2).\]
Hence by \eqref{***}
\[ \mathbf{E}^*\int_0^T\int_\mathbf{Z}\alpha^2(u,X(u^-),z)\nu(dz)du<\infty,\]
and this is a sufficient condition for the stochastic integral to be a martingale (see \cite{oksu05}, (1.1.13)).
Thus, from \eqref{**},
\begin{eqnarray*} \tilde{\Phi}(t,x)&=&\mathbf{E}^*[\tilde{\Phi}(T,X(T))\exp(Z(T))]\\
&=&\mathbf{E}^*\left[e^{\theta \int_{t}^{T}g(s,X_s,h^*_s)ds} \right]\\
&=&\tilde{I}(v,x,h^*;t;T;\theta).
\end{eqnarray*}
This completes the proof.\hfill\end{proof}\\

As usual, we have an equivalent result for the logarithmically-transformed control problem.\\

\begin{corollary}\label{Coro_JDRSAM_verification} $h^*$ defined above is optimal for the logarithmically-transformed problem, and
\[ \Phi(t,x) = I(t,x,h^*;\theta;T;v).\]
\end{corollary}

\begin{proof}
        This corollary follows from the relation between $\Phi$ and $\tilde{\Phi}$ and from the fact that an admissible (optimal) strategy for the exponentially transformed problem is also admissible (optimal) for the risk-sensitive problem.
\end{proof}

%
%

\section{Conclusion}\label{sec_conc}
In this article, we reformulated the risk-sensitive investment management problem to allow jumps in both factor levels and asset prices, stochastic volatility and investment constraints. Using a combination of viscosity solutions, change of notation, policy improvement argument and classical results on parabolic PDEs, we showed that the Bellman PIDE associated with our control problem does admit a unique smooth $C^{1,2}$ solution. Furthermore, we proved via a verification theorem that this solution and the candidate optimal control we identified solve the control problem.
\\

At first, this outcome may appear fortunate given the non-linear and non-local nature of the PDE involved. In fact, the argument used in the derivation only hinges on three techniques. First, the Lipshitz continuity of the value function provides us with the ability to rewrite the HJB PIDE as a PDE. Second, viscosity solutions give us existence and uniqueness of a weak solution to both of these equations. A proof of existence by Fleming and Rishel based on a policy improvement originally due to Bellman completes the triptic by providing a smooth solution.
\\

The robustness of the approach presented in this article is a clear advantage. Further research is needed to determine both the extend of the jump-diffusion problems this approach can be used to solve and how much further it can be developed.
\\


%
%
\appendix

%
%

\section{Proof of Proposition~\ref{prop_tildePhi_Lipschitz}}\label{App_Proof_prop_tildePhi_Lipschitz}

\textbf{Step 1:}\\
Take an admissible control $\hat{h} \in \mathcal{A(T)}$ and consider the auxiliary criterion $\tilde{J}$ defined under the $\mathbb{P}$-measure as
\begin{equation}
   \tilde{J}(v,x,\hat{h};t,T;\theta)
        = e^{-\theta J(t,x,h;\theta;T;v)}
        = \mathbf{E}_{t,x}
                \left[
        \exp \left\{ \theta \int_{t}^{s} g(s,X_s,h(s)) ds \right\} \chi^h(t)
                        \right]
                                                                                        \nonumber
\end{equation}
where the Dol\'eans exponential is defined in~\eqref{eq_JDRSAM_Doleansexp_chi}
\\

By property~\ref{prop_JDRSAM_g_bounded_Lipschitz} $g$ is bounded the constant $R$
\begin{eqnarray}\label{eq_JDRSAM_proofLipschitz_def_R}
        R := \exp\left\{ \theta \lVert g \rVert_{\infty}  (T-t) \right\}
                                                                                                                        \nonumber
\end{eqnarray}
is well defined and bounded. Let $\eta \in \mathbb{R}^n$ be a directional vector with $\lVert\eta\rVert = 1$, and $k \in \mathbb{R}$ be a scalar. Define the operator $\Delta_x F(s,x)$ as
\begin{eqnarray}
        \Delta_x F(s,x)
        := \frac{1}{k} \left[F(s,x+k\eta) - F(s,x)\right]
                                                                                                                        \nonumber
\end{eqnarray}
Hence,
\begin{eqnarray}\label{eq_JDRSAM_proofLipschitz_eq1_1}
        \Delta_x \tilde{J}(v,x,\hat{h};t,T;\theta)
        &=& \frac{1}{k} \left[\tilde{J}(v,x_1+k\eta,\hat{h};t,T;\theta) - \tilde{J}(v,x_1,\hat{h};t,T;\theta)\right]
                                                                                                                        \nonumber\\
        &=& \frac{1}{k}\mathbf{E}_{t,x_2}
                \left[
        \exp \left\{ \theta \int_{t}^{s} g(s,X_2(s),h(s)) ds \right\}
                                \chi^{\hat{h}}(t)
                        \right]
                                                                                                                        \nonumber\\
        &&      -       \frac{1}{k}\mathbf{E}_{t,x_1}
                \left[
        \exp \left\{ \theta \int_{t}^{s} g(s,X_1(s),h(s)) ds \right\}
                                \chi^{\hat{h}}(t)
                        \right]
\end{eqnarray}
where $X_1(s)$ solves the state equation~\eqref{eq_JDRSAM_state_SDE} with $X_1(t) = x_1$ and $X_2(s)$ solves the state equation~\eqref{eq_JDRSAM_state_SDE} with $X_2(t) = x_2 := x_1 + k\eta$
\\

Define
\begin{eqnarray}
    \mathcal{D}_t^i
    &:=& \exp \left\{ -\theta \int_{t}^{s} \hat{h}(s)'\Sigma(s,X_i(s)) dW_s
    -\frac{1}{2} \theta^2 \int_{t}^{s} \hat{h}(s)'\Sigma\Sigma'(s,X_i(s))\hat{h}(s) ds    \right\},
        \quad i=1,2
                                                                                                                                                                \nonumber\\
\end{eqnarray}
and
\begin{eqnarray}
    \mathcal{P}_t^h
    &:=& \exp \left\{
        +\int_{t}^{s} \int_{\mathbf{Z}} \ln\left(1-G(t,z,\hat{h}(s))\right) \tilde{N} (ds,dz)
            \right.
                                                            \nonumber\\
   &&   \left.
        +\int_{t}^{s} \int_{\mathbf{Z}} \left\{\ln\left(1-G(t,z,\hat{h}(s))\right)+G(t,z,\hat{h}(s))\right\}\nu(dz)ds
    \right\},
                                                                                                                                                                \nonumber\\
\end{eqnarray}

Let $\mathbb{P}_{\mathcal{P}}$ be the measure on
$(\Omega,\mathcal{F})$ defined via the Radon-Nikod\'ym derivative
\begin{eqnarray}
    \left. \frac{d\mathbb{P}_{\mathcal{P}}}{d\mathbb{P}}\right|_{\mathcal{F}_t}
    := \mathcal{P}_t
    \quad  \forall t \geq 0
\end{eqnarray}
then~\eqref{eq_JDRSAM_proofLipschitz_eq1_1} becomes
\begin{eqnarray}
        \Delta_x \tilde{J}(v,x,\hat{h};t,T;\theta)
        &=& \frac{1}{k}\mathbf{E}_{t,x_2}^{\mathcal{P}}
                \left[
        \exp \left\{ \theta \int_{t}^{s} g(s,X_2(s),h(s)) ds \right\}
                                \mathcal{D}_t^2
                        \right]
                                                                                                                        \nonumber\\
        &&      -       \frac{1}{k}\mathbf{E}_{t,x_1}^{\mathcal{P}}
                \left[
        \exp \left\{ \theta \int_{t}^{s} g(s,X_1(s),h(s)) ds \right\}
                                \mathcal{D}_t^1
                        \right]
                                                                                                                        \nonumber
\end{eqnarray}
where $\mathbf{E}_{t,x}^{\mathcal{P}} \left[ \cdot \right]$ denotes the
expectation taken with respect to the measure $\mathbb{P}^{\mathcal{P}}$ and with initial conditions $(t,x)$.
\\

Thus
\begin{eqnarray}
        \Delta_x \tilde{J}(v,x,\hat{h};t,T)
        &=& \frac{1}{k}\mathbf{E}_{t,x_2}^{\mathcal{P}}
                \left[
        \exp \left\{ \theta \int_{t}^{s} g(s,X_2(s),h(s)) ds \right\}
                                \mathcal{D}_t^2
                        \right]
                                                                                                                        \nonumber\\
        &&      -       \frac{1}{k}\mathbf{E}_{t,x_1}^{\mathcal{P}}
                \left[
        \exp \left\{ \theta \int_{t}^{s} g(s,X_1(s),h(s)) ds \right\}
                                \mathcal{D}_t^2\frac{\mathcal{D}_t^1}{\mathcal{D}_t^2}
                        \right]
                                                                                                                        \nonumber
\end{eqnarray}

Let $\mathbb{P}_{2}$ be the measure on
$(\Omega,\mathcal{F})$ defined via the Radon-Nikod\'ym derivative
\begin{eqnarray}
    \left. \frac{d\mathbb{P}_{2}}{d\mathbb{P}_{\mathcal{P}}}\right|_{\mathcal{F}_t}
    := \mathcal{D}_t^2
    \quad  \forall t \geq 0
\end{eqnarray}
Then,
\begin{eqnarray}\label{eq_JDRSAM_proofLipschitz_eq1_2}
        \Delta_x \tilde{J}(v,x,\hat{h};t,T;\theta)
        &=& \frac{1}{k}\mathbf{E}_{t,x_2}^{2}
                \left[
                \exp \left\{ \theta \int_{t}^{s} g(s,X_2(s),h(s)) ds \right\}
                -       \exp \left\{ \theta \int_{t}^{s} g(s,X_1(s),h(s)) ds \right\}
                        \right]
                                                                                                                        \nonumber\\
        &&      -       \frac{1}{k}\mathbf{E}_{t,x_2}^{2}
                \left[
        \exp \left\{ \theta \int_{t}^{s} g(s,X_1(s),h(s)) ds \right\}
                                \left(\mathcal{D}_t^{1,2}-1\right)
                        \right]
                                                                                                                        \nonumber\\
        &=&     \frac{1}{k}\left(A + B\right)
\end{eqnarray}
where
\begin{eqnarray}
        A:= \mathbf{E}_{t,x_2}^{2}
                \left[
                \exp \left\{ \theta \int_{t}^{s} g(s,X_2(s),h(s)) ds \right\}
                -       \exp \left\{ \theta \int_{t}^{s} g(s,X_1(s),h(s)) ds \right\}
                        \right]
                                                                                                                        \nonumber
\end{eqnarray}
\begin{eqnarray}
        B := \mathbf{E}_{t,x_2}^{2}
                \left[
        \exp \left\{ \theta \int_{t}^{s} g(s,X_1(s),h(s)) ds \right\}
                                \left(\mathcal{D}_t^{1,2}-1\right)
                        \right]
                                                                                                                        \nonumber
\end{eqnarray}
\begin{eqnarray}
        \mathcal{D}_t^{1,2}
        &:=& \frac{\mathcal{D}_t^1}{\mathcal{D}_t^2}
                                                                                                                        \nonumber\\
        &=& \exp \left\{ -\theta \int_{t}^{s} \hat{h}(s)'
                                \left[\Sigma(s,X_1(s))-\Sigma(s,X_2(s))\right] dW_s^2
                \right.
                                                                                                                        \nonumber\\
        &&      \left.
    -\frac{1}{2} \theta^2 \int_{t}^{s} \hat{h}(s)'
                        \left[\Sigma(s,X_1(s))-\Sigma(s,X_2(s))\right]
                        \left[\Sigma(s,X_2(s))-\Sigma(s,X_2(s))\right]'
                \hat{h}(s) ds    \right\}
                                                                                                                        \nonumber
\end{eqnarray}
and $W^2(s)$ is a $\mathbb{P}_{2}$-Brownian motion.
\\

\textbf{Step 2}\\
We focus on term $B$ in~\eqref{eq_JDRSAM_proofLipschitz_eq1_2}. By~\eqref{eq_JDRSAM_proofLipschitz_def_R}, we have
\begin{eqnarray}
        B \leq R \mathbf{E}_{t,x_2}^{2}
                \left[
                                \left(\mathcal{D}_t^{1,2}-1\right)
                        \right]
                                                                                                                        \nonumber
\end{eqnarray}
The exponential martingale $\mathcal{D}_t^{1,2}$ satisfies the SDE
\begin{eqnarray}
        \frac{d\mathcal{D}^{1,2}(s)}{\mathcal{D}^{1,2}(s)}
        &=& -\theta \hat{h}(s)'
                                \left[\Sigma(s,X_1(s))-\Sigma(s,X_2(s))\right] dW_s^2
                                                                                                                        \nonumber
\end{eqnarray}
thus,
\begin{eqnarray}
        \mathcal{D}^{1,2}(s)    - 1
        &=& -\theta \int_{t}^{s} \mathcal{D}^{1,2}(u) \hat{h}(u)'
                                \left[\Sigma(u,X_1(u))-\Sigma(u,X_2(u))\right] dW_u^2,
                                                                                                                        \nonumber\\
        &=& Y_1(s) - Y_2(s)
                                                                                                                        \nonumber
\end{eqnarray}
with
\begin{eqnarray}
        Y_i(s)
        :=
        -\theta \int_{t}^{s} \mathcal{D}^{1,2}(u) \hat{h}(u)'\Sigma(u,X_i(u))dW_u^2
        \qquad i=1,2
                                                                                                                        \nonumber
\end{eqnarray}
\\

Now,
\begin{eqnarray}
        &&      \mathbf{E}_{t,x_2}^{2}\left[ |Y_1(s) - Y_2(s)|^2 \right]
                                                                                                                        \nonumber\\
        &=& \mathbf{E}_{t,x_2}^{2}\left[ \int_{t}^{s}
                        \left(\mathcal{D}^{1,2}(u)\right)^2
                        \hat{h}(u)'
                                \left[\Sigma(u,X_1(u))-\Sigma(u,X_2(u))\right]
                                \left[\Sigma(u,X_1(u))-\Sigma(u,X_2(u))\right]'\hat{h}(u)
                        du \right]
                                                                                                                        \nonumber
\end{eqnarray}
By Cauchy-Schwartz,
\begin{eqnarray}\label{eq_JDRSAM_proofLipschitz_eq3_1}
        &&      \mathbf{E}_{t,x_2}^{2}\left[ |Y_1(s) - Y_2(s)|^2 \right]
                                                                                                                        \nonumber\\
        &\leq& \int_{t}^{s}\left\{\sqrt{\mathbf{E}_{t,x_2}^{2}\left[
                        \left(\mathcal{D}^{1,2}(u)\right)^4\right]}
                        \right.
                                                                                                                        \nonumber\\
        &&              \left.
                        \sqrt{\mathbf{E}_{t,x_2}^{2}\left[
                        \left(\hat{h}(u)'
                                \left[\Sigma(u,X_1(u))-\Sigma(u,X_2(u))\right]
                                \left[\Sigma(u,X_1(u))-\Sigma(u,X_2(u))\right]'\hat{h}(u)\right)^2
                        \right]}\right\}
                        du
                                                                                                                        \nonumber\\
\end{eqnarray}
\\

The term
\begin{eqnarray}
        \left(\mathcal{D}_u^{1,2}\right)^4
        &:=& \left(\frac{\mathcal{D}_u^1}{\mathcal{D}_u^2}\right)^4
                                                                                                                        \nonumber\\
        &=& \exp \left\{ -4\theta \int_{t}^{s} \hat{h}(r)'
                                \left[\Sigma(r,X_1(r))-\Sigma(r,X_2(r))\right] dW_r^2
                \right.
                                                                                                                        \nonumber\\
        &&      \left.
    -2 \theta^2 \int_{t}^{s} \hat{h}(s)'
                        \left[\Sigma(r,X_1(r))-\Sigma(r,X_2(r))\right]
                        \left[\Sigma(r,X_2(r))-\Sigma(r,X_2(r))\right]'
                \hat{h}(r) dr    \right\}
                                                                                                                        \nonumber\\
        &=& \mathcal{D}_u^{3}\exp \left\{
        6 \theta^2 \int_{t}^{u} \hat{h}(r)'
                        \left[\Sigma(r,X_1(r))-\Sigma(r,X_2(r))\right]
                        \left[\Sigma(r,X_2(r))-\Sigma(r,X_2(r))\right]'
                \hat{h}(r) dr    \right\}
                                                                                                                        \nonumber
\end{eqnarray}
where
\begin{eqnarray}
        \mathcal{D}_u^{3}
        &:=& \exp \left\{ -4\theta \int_{t}^{u} \hat{h}(r)'
                                \left[\Sigma(r,X_1(r))-\Sigma(r,X_2(r))\right] dW_r^2
                \right.
                                                                                                                        \nonumber\\
        &&      \left.
    - 8\theta^2 \int_{t}^{u} \hat{h}(r)'
                        \left[\Sigma(r,X_1(r))-\Sigma(r,X_2(r))\right]
                        \left[\Sigma(r,X_2(r))-\Sigma(r,X_2(r))\right]'
                \hat{h}(r) dr    \right\}
                                                                                                                        \nonumber
\end{eqnarray}
Let $\mathbb{P}_{3}$ be the measure on
$(\Omega,\mathcal{F})$ defined via the Radon-Nikod\'ym derivative
\begin{eqnarray}
    \left. \frac{d\mathbb{P}_{3}}{d\mathbb{P}_{2}}\right|_{\mathcal{F}_t}
    := \mathcal{D}_t^3
    \quad  \forall t \geq 0
\end{eqnarray}
then
\begin{eqnarray}\label{eq_JDRSAM_proofLipschitz_eq3_2}
        &&      \sqrt{\mathbf{E}_{t,x_2}^{2}\left[
                        \left(\mathcal{D}^{1,2}(u)\right)^4\right]}
                                                                                                \nonumber\\
        &=&
        \sqrt{\mathbf{E}_{t,x_2}^{3}\left[
                        \exp \left\{
        6 \theta^2 \int_{t}^{u} \hat{h}(r)'
                        \left[\Sigma(r,X_1(r))-\Sigma(r,X_2(r))\right]
                        \left[\Sigma(r,X_2(r))-\Sigma(r,X_2(r))\right]'
                \hat{h}(r) dr    \right\}
        \right]}
                                                                                                \nonumber\\
        &\leq&  K_3
\end{eqnarray}
where $\mathbf{E}_{t,x}^{3}$ is the expectation with respect to the measure $\mathbb{P}_{3}$. The last inequality follows by boundedness of $\Sigma(\cdot)$ and the constant $K_3$ depends on the choice of control $\hat{h}$.
\\

We now look at the second term on the right hand side of~\eqref{eq_JDRSAM_proofLipschitz_eq3_1}. By Lipschitz continuity,
\begin{eqnarray}
        &&
\left(\hat{h}(u)'\left[\Sigma(u,X_1(u))-\Sigma(u,X_2(u))\right]
        \left[\Sigma(u,X_1(u))-\Sigma(u,X_2(u))\right]'\hat{h}(u)\right)^2
                                                                                        \nonumber\\
        &\leq&
        K_4 |X_1(u)-X_2(u)|^4
                                                                                        \nonumber
\end{eqnarray}
for $u \in [t,s]$ and where the constant $K_4>0$ depends on $\hat{h}$ and the Lipschitz constant $K_{\Sigma}$.
\\

Let $Z(u) = X_1(u)-X_2(u)$, then
\begin{eqnarray}
    dZ(r)
    &=&     f(r,Z(r)dr
            + \Lambda(r,Z(r)) dW_{r}^{\theta}
                                + \int_{\mathbf{Z}} \xi(z,Z(r))
                                                \tilde{N} ^{h}(dr,dz)
                                                                                                                \nonumber
\end{eqnarray}
where
\begin{eqnarray}
        f(r,Z(r))       &=&     f(r,X_1(r),h(r)) - f(r,X_2(r),h(r))
                                                                                                                \nonumber\\
        \Lambda(r,Z(r)) &=&     \Lambda(r,X_1(r)) - \Lambda(r,X_2(r))
                                                                                                                \nonumber\\
        \xi(z,Z(r))
        &=&     \xi\left(r,X_1(r^-),z\right)    - \xi\left(r,X_2(r^-),z\right)
                                                                                                                \nonumber
\end{eqnarray}
\\

By It\^o,
\begin{eqnarray}
        |Z(u)|^4
        &=&     |Z(t)|^4
        +       4\int_{t}^{u} Z^3(r)f(r,Z(r)) dr
        +       4\int_{t}^{u} Z^3(r)\Lambda(r,Z(r)) dW(r)
                                                                                                \nonumber\\
        &&
        +       6\int_{t}^{u} Z^2(r)\Lambda\Lambda'(r,Z(r)) dr
                                                                                                \nonumber\\
        &&
        + \int_{t}^{u}\int_{\mathbf{Z}} \left\{
                (Z(u^-)+\xi(z,Z(u^-)))^4 - Z^4(u^-) - 4 Z^3(u^-)\cdot\xi(z,Z(u))
                \right\} \nu(dz)dr
                                                                                                \nonumber\\
        &&
        + \int_{t}^{u}\int_{\mathbf{Z}} \left\{
                (Z(u^-)+\xi(z,Z(u^-)))^4 - Z^4(u^-) - 4 Z^3(u^-)\cdot\xi(z,Z(u))
                \right\} \tilde{N} (dz,dr)
                                                                                                        \nonumber
\end{eqnarray}

Taking the expectation under $\mathbb{P}_2$,
\begin{eqnarray}
        \mathbf{E}_{t,x_2}^{2}|Z(u)|^4
        &=&     |Z(t)|^4
        +       2\mathbf{E}_{t,x_2}^{2}\left[\int_{t}^{u} Z^3(r)f(r,Z(r))
        +       3\int_{t}^{u} Z^2(r)\Lambda\Lambda'(r,Z(r)) dr \right]
                                                                                                \nonumber\\
        &&
        + \mathbf{E}_{t,x_2}^{2}\left[\int_{t}^{u}\int_{\mathbf{Z}} \left\{
                (Z(u^-)+\xi(z,Z(u^-)))^4 - Z^4(u^-) - 4 Z^3(u^-)\cdot\xi(z,Z(u))
                \right\} \nu(dz)dr \right]
                                                                                                \nonumber
\end{eqnarray}
\\

Applying Assumption~\ref{As_Factors} (\ref{As_Factors_i}), (\ref{As_Factors_ii}), (\ref{As_Factors_v}) and Assumption~\ref{As_Securities} (\ref{As_Securities_ix}) as well as Remark~\ref{rk_JDRSAM_f_Lipschitz} we finally obtain
\begin{eqnarray}
        \mathbf{E}_{t,x_2}^{2}|X_1(u)-X_2(u)|^4
        &=&     |x_1-x_2|^4
        +       C_4 \int_{t}^{u}\mathbf{E}_{t,x_2}^{2} |X_1(r)-X_2(r)|^4 dr
                                                                                                \nonumber
\end{eqnarray}
\begin{eqnarray}
        C_0
        := 4(K_b + \theta K_{\Lambda\Sigma}+K_{\xi}K_0)
        +       6K_{\Lambda}^2
        +\int_{\mathbf{Z}}K_{\xi}^2(6+4K_{\xi}+K_{\xi}^2)\nu(dz)
                                                                                                                                                \nonumber
\end{eqnarray}
and
\begin{eqnarray}
        K_0
        := \int_{\mathbf{Z}}
                        \left[ \left(1+h'\gamma(t,z)\right)^{-\theta}
                        - \mathit{1}_{\mathbf{Z}_0}(z)
                \right]\nu(dz)
                                                                                                                                                \nonumber
\end{eqnarray}
By Gronwall's inequality,
\begin{eqnarray}
        \mathbf{E}_{t,x}^{\hat{h},\theta} |X_1(u)-X_2(u)|^2
        \leq e^{C_4(u-t)}|x_1-x_2|^2
                                                                                                                                                \nonumber
\end{eqnarray}
\\

Thus, taking into consideration~\eqref{eq_JDRSAM_proofLipschitz_eq3_2}, equation~\eqref{eq_JDRSAM_proofLipschitz_eq3_1} becomes
\begin{eqnarray}
        \mathbf{E}_{t,x_2}^{2}\left[ |Y_1(s) - Y_2(s)|^2 \right]
        &\leq& e^{C_4(u-t)}|x_1-x_2|^2
                                                                                                                        \nonumber
\end{eqnarray}
and we conclude that term $B$ in~\eqref{eq_JDRSAM_proofLipschitz_eq1_2} is Lipschitz continuous since
\begin{eqnarray}
        B \leq C_B|x_1-x_2|
                                                                                                                        \nonumber
\end{eqnarray}
where $C_B := K_3R e^{\frac{1}{2}C_4(u-t)}$
\\

\textbf{Step 3: Conclusion}\\
The Lipschitz continuity of term $A$ in~\eqref{eq_JDRSAM_proofLipschitz_eq1_2} follows from Theorem VI.8.1 and Lemma IV.7.1 in~\cite{flso06}. Let $C_A>0$ be the Lipschitz constant for term $A$, then
\begin{eqnarray}
        &&      \Delta_x \tilde{J}(v,x,\hat{h};t,T)
                                                                                                                        \nonumber\\
        &=& \frac{1}{k}\mathbf{E}_{t,x_2}^{2}
                \left[
                \exp \left\{ \theta \int_{t}^{s} g(s,X_2(s),h(s)) ds \right\}
                -       \exp \left\{ \theta \int_{t}^{s} g(s,X_1(s),h(s)) ds \right\}
                        \right]
                                                                                                                        \nonumber\\
        &&      -       \frac{1}{k}\mathbf{E}_{t,x_2}^{2}
                \left[
        \exp \left\{ \theta \int_{t}^{s} g(s,X_1(s),h(s)) ds \right\}
                                \left(\mathcal{D}_t^{1,2}-1\right)
                        \right]
                                                                                                                        \nonumber\\
        &\leq&  C_x |x_2 - x_1|
                                                                                                                        \nonumber
\end{eqnarray}
where $C_x := C_B + C_A$. This completes the proof.

%
%

\end{document}